\documentclass[floatfix,aip,cha,10pt,twocolumn,superscriptaddress,longbibliography]{revtex4-2}

\usepackage[utf8]{inputenc}
\usepackage[T1]{fontenc}
\usepackage[colorlinks=true,linkcolor=blue,urlcolor=blue,citecolor=blue,breaklinks=true]{hyperref}

\usepackage{amsmath,amsthm,amssymb,amscd,bbm,mathtools}
\mathtoolsset{showonlyrefs}

\newcommand{\eye}{\mbox{$\mbox{1}\!\mbox{l}\;$}}
\renewcommand{\vec}[1]{\boldsymbol{#1}}
\newcommand{\matr}[1]{\boldsymbol{#1}}

\newtheorem{lemma}{Lemma}
\newtheorem{cor}{Corollary}

\begin{document}

\title{Dynamic stability of electric power grids: Tracking the interplay of the network structure, transmission losses and voltage dynamics}

\author{Philipp C. B\"ottcher~\href{https://orcid.org/0000-0002-3240-0442}{\includegraphics[width=3.2mm]{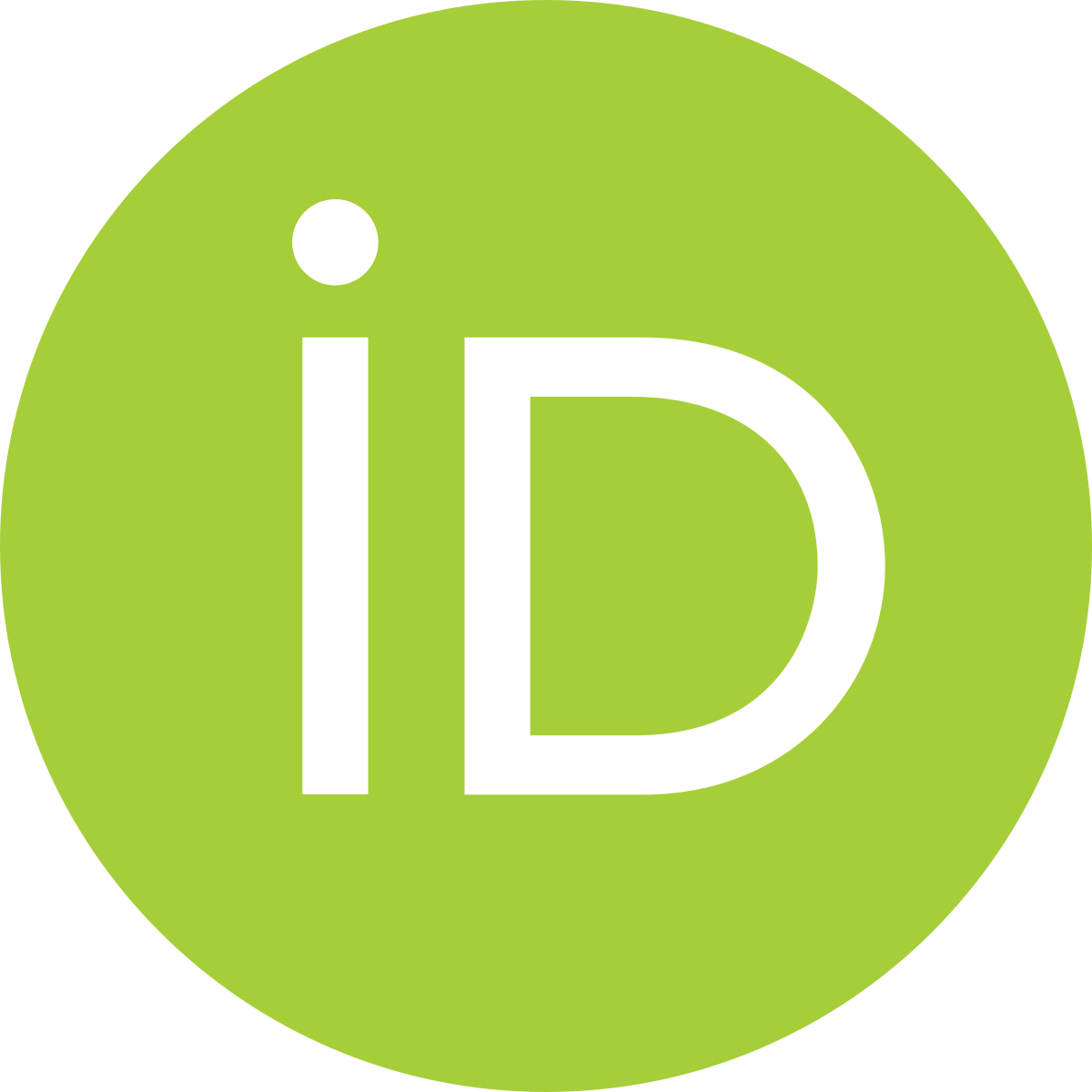}}}
\email[Electronic mail: ]{p.boettcher@fz-juelich.de}
\affiliation{Forschungszentrum J\"ulich, Institute for Energy and Climate Research -\\
Systems Analysis and Technology Evaluation (IEK-STE), 52428 J\"ulich, Germany\looseness=-1}
\affiliation{German Aerospace Center (DLR), Institute of Networked Energy Systems, Oldenburg, Germany\looseness=-1}

\author{Dirk~Witthaut~\href{https://orcid.org/0000-0002-3623-5341}{\includegraphics[width=3.2mm]{orcid.png}}}
\email[Electronic mail: ]{d.witthaut@fz-juelich.de}
\affiliation{Forschungszentrum J\"ulich, Institute for Energy and Climate Research -\\
Systems Analysis and Technology Evaluation (IEK-STE), 52428 J\"ulich, Germany\looseness=-1}
\affiliation{Institute for Theoretical Physics, University of Cologne, 50937 K\"oln, Germany\looseness=-1}

\author{Leonardo~Rydin~Gorj\~ao~\href{https://orcid.org/0000-0001-5513-0580}{\includegraphics[width=3.2mm]{orcid.png}}}
\email[Electronic mail: ]{leonardo.rydin@gmail.com}
\affiliation{Forschungszentrum J\"ulich, Institute for Energy and Climate Research -\\
Systems Analysis and Technology Evaluation (IEK-STE), 52428 J\"ulich, Germany\looseness=-1}
\affiliation{German Aerospace Center (DLR), Institute of Networked Energy Systems, Oldenburg, Germany\looseness=-1}
\affiliation{Institute for Theoretical Physics, University of Cologne, 50937 K\"oln, Germany\looseness=-1}

\begin{abstract}
Dynamic stability is imperative for the operation of the electric power system.
This article provides analytical results and effective stability criteria focusing on the interplay of network structures and the local dynamics of synchronous machines.
The results are based on an extensive linear stability analysis of the third-order model for synchronous machines, comprising the classical power-swing equations and the voltage dynamics.
The article explicitly covers the impact of Ohmic losses in a linear approximation in power grids, which are often neglected in analytical studies.
Necessary and sufficient stability conditions are formulated, and different routes to instability are analysed, yielding concrete mathematical criteria applicable to all scales of power grids, from transmission to distribution grids, as well as microgrids.
A subsequent numerical study of the criteria is presented, without and with resistive terms, to test how tight the derived analytical results are.
\end{abstract}

\maketitle

\begin{quotation}
The secure supply of electric power relies on the stable, coordinated operation of thousands of electric machines connected via the electric power grid. 
At the transmission grid level, machines run synchronously with fixed voltage magnitudes and stationary relative phase angles defining a stationary state. 
The ongoing introduction of renewable power systems poses several challenges to the stability of the system, as situations with highly loaded lines and temporal fluctuations increase considerably. 
This trend takes place in both the transmission grid at high voltages, as well as in distribution grids and microgrids at medium and low voltages. 
This article contributes to the understanding of dynamical stability of electric power systems and provide a detailed analysis of the third-order model for synchronous generators, which includes the transient dynamics of voltage magnitudes.
Special emphasis is laid on the impact of Ohmic losses in the transmission of power, which are often neglected in analytical treatments of power system stability. 
The analytical results thus find applicability on all size scales of power grids, from transmission grids to isolated microgrids, for openly tackling systems with losses in a rigorous analytical manner.
Furthermore, the results are independent of the network construction and entail explicit criteria for the connectivity of the power grid and the physical requirements needed to ensure stability in the presence of resistive terms.
\end{quotation}

\section{Introduction}\vspace*{-2.5mm}
The unwavering operation of the electric power systems is vital to our daily life and the continued function of modern societies as a whole.
Thus, an improved understanding of the electrical system's dynamic properties is especially relevant at present, as more renewable energies enter the electric power-grid systems~\cite{Sims11,Milano2018}.
One of the key elements at play is the relative reduction of inertial mass in power systems due to the penetration of renewable generation, which can lead to large dynamic responses to disturbances~\cite{Schaefer2018,Boettcher2020,Markovic2021}.
An expected higher grid load and stronger fluctuating generation by wind and solar resources may further threaten the dynamic stability~\cite{Farmer2021a,Farmer2021b,RydinGorjao2021}. 

Conventional power generation typically involves large rotating masses that offer stabilising inertia~\cite{Machowski2020}.
These power grids can be of various scales, spanning entire continents to single islands.
Recently, the concept of microgrids has emerged~\cite{Lasseter2004,Rocabert2012,Guo2017}: partially independent power grids in smaller environments that are coupled to a main power grid.
These power grids operate at lower voltages than conventional transmission grids and are capable of producing their own power, consequently working partially independently from an overlying power grid~\cite{Tayab2017,Doerfler2019}.
Microgrids embedded in a power grids are still ruled by a common understanding of fixed nominal frequency, e.g., $50$~Hz in Europe, among many other stability criteria~\cite{Rohden2014,FarokhianFiruzi2019}.

Analytical approaches to stability in power-grid systems are a difficult task and generally rely on model simplifications to keep the problem tractable~\cite{Pagnier2019,Tyloo2019,Tyloo2019b,Hellmann2020,Thuemler2021}.
The most common simplification to make stability problems mathematically tractable is the assumption of having lossless systems~\cite{He2019,Qiu2020}.
Various such studies with complex dynamical models exist, cf. Schiffer \textit{et al.} \cite{Schiffer2014,Schiffer2015,Schiffer2017} and D{\"o}rfler \textit{et al.} \cite{Doerfler2010,Doerfler2013}, yet results are scarce for extended networks including resistive terms, given the difficulty of tackling dissipative systems mathematically.
The problem of losses in power-grid systems is often tackled using extensive numerical simulations~\cite{Krause2002,Sauer2017}.

This article puts forth a set of mathematical stability criteria for power grids based on the third-order model for synchronous generators~\cite{Schmietendorf2014,Auer2016,Schmietendorf2016,Pierre2019,Wu2020,Suchithra2021}.
The criteria can be employed for various scales of power grids -- for both transmission and distribution grids -- evidencing the limitations entailed by the existence of resistive terms on the operability of power-grid systems.
In particular, the article undertakes the task of intertwining results for graph theory with the characteristics of the power-grid construction and their physical properties~\cite{Huang2020,Wenting2021}, extending a previous article on lossless power grids~\cite{Sharafutdinov2018}.

The article is structured in the following manner: Section~\ref{sec:II} introduces the basic dynamical model studied in this paper.
In this work, we present an analysis of the third-order model, comprising transient voltage dynamics and considering extended grids with complex topology and resistive losses.
Section~\ref{sec:III} tackles the linear stability analysis of the equations of motion, a reduction of the problem to a matrix formulation and develops a mathematical apparatus to unveil sufficient and necessary criteria for stability in a general sense.
Section~\ref{sec:IV} introduces the two main lemmata of the article from which various stability criteria are derived.
Lemma~\ref{lemma:Schur-lossless} covers solely lossless grids and Lemma~\ref{lemma:Schur-lossy} extends the results to the case of lossy transmission up to leading order in the losses in the system, which is drawn from perturbation theory.
In Section~\ref{sec:V} the developed concepts are utilised to derive analytical stability conditions for both lossless and lossy systems, presenting criteria for stability not only for the power-angle and the voltage dynamics, but also for a mixed type of instabilities.
These results also represent a direct link with graph-theoretical measures.
Section~\ref{sec:VI} comprises a set of numerical studies on model systems to check how tight tight the derived bounds are.
The conclusions follow subsequently in Section~\ref{sec:VII}.

\section{Modelling scale-independent network-based power grids}\label{sec:II}\vspace*{-2.5mm}
\subsection{Third-order model for synchronous generators}\vspace*{-2.5mm}

The third-order model for synchronous machines, denoted as well as a one- or $q$-axis model, describes the transient dynamics of coupled synchronous machines \cite{Krause2002,Sauer2017,Machowski2020}.
It embodies the power- or rotor angle $\delta(t)$, relative to the power-grid reference frame, the angular frequency $\omega(t) = \dot{\delta}(t)$, in a co-rotating reference frame rotating with the reference frequency $\Omega$, and the transient voltage $E_q(t)$, in the $q$-direction of a co-rotating frame of reference of each machine in the system.
It excludes sub-transient effects, i.e., higher-order effects, and assumes that the transient voltage $E_d$ in the $d$-direction of the co-rotating frame vanishes.

Sub-transient effects play a small role, especially in the case of studying power grids in the vicinity of the steady state~\cite{Weckesser2013}.
The truncation of the transient voltage $E_d$ in the $d$-axis is imposed out of necessity to have an analytically tractable model.
Still, the resulting dynamical system is rather complex such that analytical results are scarce and mostly restricted to lossless power grids.
Hence, the scope of the analysis here is two-fold: To present the details of tackling rotor-angle and voltage stability, whilst not shunning away from complex network topologies and considering Ohmic losses explicitly.

The equations of motion for one generator are given by~\cite{Machowski2020}
\begin{equation}
\begin{aligned}\label{eq:3rdorder-pre}
   \dot \delta &= \omega, \\
   M \dot \omega &=   - D \omega + P^{\text{m}} - P^{\text{el}}, \\
   T \dot E &= E^{f} - E + (X - X') I,
  \end{aligned}
\end{equation} 
where henceforth $E \equiv E_q$ denotes solely the voltage along the $q$-axis, and the dot the differentiation with respect to time.
Furthermore, $P^{\mathrm{}{m}}$ denotes the effective mechanical input power of the machine, $E^{f}$ the internal voltage or field flux, and $P^{\mathrm{el}}$ denotes the electrical power out-flow.
The parameters $M$ and $D$ are the inertia and damping of the mechanical motion and $T$ the relaxation time of the transient voltage dynamics.
The voltage dynamics further depend on the difference of the static reactance $X$ and transient reactance $X'$ along the $d$-axis, where $X - X' > 0$ in general, and the current $I$ along $d$-axis. 

The active electrical power $P^{\mathrm{el}}_j$ exchanged with the power grid, and the current $I_{j}$ at the $j$-th machine read~\cite{Schmietendorf2014}
\begin{equation}
\begin{aligned}
    P^{\mathrm{el}}_j &= \sum_{\ell=1}^N\! E_{j} E_{\ell} \left[B_{j,\ell} \sin(\delta_j\!-\!\delta_\ell ) \!+\! G_{j,\ell} \cos(\delta_j\!-\!\delta_\ell )\right], \\
    I_{j}\! &=\! \sum_{\ell=1}^N\! E_{\ell} \left[B_{j,\ell} \cos(\delta_j\!-\!\delta_\ell )\! -\! G_{j,\ell} \sin(\delta_j\!-\!\delta_\ell )\right],
\end{aligned}
\end{equation}
where the $E_{j}$ and $\delta_j$ are the transient voltage and the rotor angle of the $j$-th machine, respectively.
The parameters $G_{j,\ell}$ and $B_{j,\ell}$ denote the real and imaginary parts of the nodal admittance matrix and encode the network structure.
Generally, $B_{j,\ell}>0$ and $G_{j,\ell}<0$ for all $j\neq \ell$.
This article is especially concerned with the role of Ohmic losses, which are described by the real parts of the nodal admittance matrix $G_{j,\ell}$.
All quantities are usually made dimensionless using appropriate scaled units referred to as the `pu system' or `per unit system' \cite{Machowski2020}.

Load nodes are typically described by constant impedances to the ground.
These passive nodes can be eliminated from the network equations via Kron reduction such that only generator nodes have to be considered explicit~\cite{Kron1939,Doerfler2011b}.
The remaining nodes are then connected by an effective network which differs considerably from the physical one.
For instance, the reduced network is typically fully connected.

The equations of motion~\eqref{eq:3rdorder-pre} for the $j$-th synchronous machine, in a system with $N$ machines, take the form~\cite{Schmietendorf2014,Jinpeng2016,Schmietendorf2016}
\begin{equation}
\begin{aligned}\label{eq:3rdorder}
   \dot \delta_j\!&=\!\omega_j,\\
 M_j \dot \omega_j\!&=\!P^{m}_j\!-\!D_j \omega_j\!+\sum_{\ell=1}^N\! E_j E_\ell \left[B_{j,\ell} \sin(\delta_\ell\!-\!\delta_j )\right. \\
                  & \qquad\qquad\qquad\qquad\qquad\quad   \left.+ G_{j,\ell} \cos(\delta_\ell\!-\!\delta_j )\right], \\   
   T_j \dot E_j\!&=\!E_j^{f}\!-\!E_j\!+\!(X_j\!-\!X'_j)\!\sum_{\ell=1}^N\!E_{\ell} \left[B_{j,\ell} \cos(\delta_j\!-\!\delta_\ell )\right. \\
                  & \qquad\qquad\qquad\qquad\qquad\quad   - \left. G_{j,\ell} \sin(\delta_j\!-\!\delta_\ell )\right].
\end{aligned}
\end{equation}


Most analytical studies so far neglected, under reasonable assumptions, the line losses of the power-grid structure.
The terms proportional to $G_{j,\ell}$ are assumed negligible in comparison to the terms proportional to $ B_{j,\ell}$.
Such arguments are reasonable for the high-voltage transmission grid, but are mostly unfounded for distribution and microgrids, where the resistance and inductance of transmission lines are comparable~\cite{Rocabert2012}.
In addition, losses become more considerable in magnitude when the transmitted power is large. 
This manuscript puts forth a study of the system in full form, not discarding the interplay of susceptance and conductance, i.e., fully integrating losses, by taking a perturbation theory approach to the losses.

\section{Equilibria and linear stability analysis}\label{sec:III}\vspace*{-2.5mm}
\subsection{Equilibrium states of power grid operation}\vspace*{-2.5mm}

The stationary operation of the voltages and power-angles of the machines comprising the power grid is the cornerstone of operability of power grids.
Constant voltages and perfect phase-locking, i.e., a point in configuration space where all $E_j$, $\omega_j$ and $\delta_j - \delta_\ell$ are constant in time, is the desired state.
The latter restriction requires that all machines rotate at the same frequency $\delta_j(t) = \Omega t + \delta_j^\circ$ for all $j=1,\ldots,N$, leading to the conditions
\begin{equation}\label{eq:equi_point}
   \dot \omega_j = \dot E_j = 0, ~ \dot \delta_j =  \Omega, \quad \forall j=1,\ldots,N.
\end{equation}
In dynamical system terms, this is a stable limit cycle of the system, also known as an isolated closed orbit.
From a physical perspective, all points on the limit cycle are equivalent as they only differ by a global phase $\alpha$ which is irrelevant for the operation of the power grid.
One can thus choose one of these points as a representative of the limit cycle and refer to it as an `equilibrium manifold'.
The superscript $\cdot^\circ$ is used to denote the values of the rotor-phase angle, frequency, and voltage in this equilibrium manifold.
Likewise, perturbations along the limit cycle do not affect the power grid operation and can thus be excluded from the stability analysis.

For the third-order model~\eqref{eq:3rdorder} an equilibrium manifold of the power grid is given by the nonlinear algebraic equations
\begin{equation}\label{eq:3rd-fixed}
\begin{aligned}
   \Omega &= \omega_j^{\circ},  \\
   0 &=  P^{m}_j -D_j \Omega  +\sum_{\ell=1}^N E_j^{\circ} E_\ell^{\circ} \left[B_{j,\ell} \sin(\delta_\ell ^\circ- \delta_j^\circ )\right. \\
     & \qquad\qquad\qquad\qquad\qquad\qquad \left.+ G_{j,\ell} \cos(\delta_\ell^\circ- \delta_j^\circ )\right],\\
   0 &= E_j^{f}\!-\!E_j^\circ\!+\!(X_j\!-\!X'_j) \sum_{\ell=1}^N E^\circ_{\ell} \left[B_{j,\ell} \cos(\delta^\circ_j- \delta^\circ_\ell )\right. \\
     & \qquad\qquad\qquad\qquad\qquad\qquad - \left. G_{j,\ell} \sin(\delta^\circ_j- \delta^\circ_\ell )\right],
\end{aligned}    
\end{equation}
noting that many equilibria -- stable and unstable -- can exist in networks with sufficiently complex topology, although this does not preclude performing a linear stability analysis~\cite{Korsak1972,Delabays2016,Manik2017,Jafarpour2019,Hellmann2020}.

\subsection{Linear stability analysis}\vspace*{-2.5mm}
A central tool of dynamical systems study is linear or small-signal stability analysis~\cite{Strogatz2015}.
The local stability properties of an equilibrium $(\delta_j^{\circ},\omega_j^{\circ},E_j^{\circ})$, i.e., stability with respect to small perturbations around an equilibrium point, can be obtained by linearising the equations of motion of the system~\eqref{eq:3rdorder}.

To perform a linear stability analysis of~\eqref{eq:3rdorder}, one introduces the perturbations $\xi_j$, $\nu_j$ and $\epsilon_j$, such that
\begin{equation}
   \delta_j(t) = \delta_j^{\circ} + \xi_j(t),  \;  \omega_j(t) =  \omega_j^{\circ} + \nu_j(t), \; E_j(t) = E_j^{\circ} + \epsilon_j(t).
\end{equation}
The rotor-angle perturbation $\xi_j$, the frequency perturbation $\nu_j$, and the voltage perturbation $\epsilon_j$ can, individually or collectively, decay to zero or grow indefinitely.
This on the other hand does not exclude the existence of other attractors in state space but the linearisation around a fixed point will only preserve the attractor around the given fixed point.
The system, around the equilibrium $(\delta_j^{\circ},\omega_j^{\circ},E_j^{\circ})$, is either stable or unstable, correspondingly.
This is also known as `exponential stability' or `small-signal stability'.

Applying the linearisation of~\eqref{eq:3rdorder} whilst simultaneously gauging onto a rotating frame of reference, with rotation frequency $\Omega$ as in~\eqref{eq:equi_point}, yields
\begin{align}\label{eq:linstab1}
    \dot \xi_j &= \nu_j, \nonumber \\
    M_j \dot \nu_j &=\!-D_j \nu_j\!-\!\sum_{\ell=1}^N (\Lambda_{j,\ell}\!+\!\Gamma_{j,\ell}) \xi_\ell  + \sum_{\ell=1}^N (A_{\ell,j}\!+\!C_{j,\ell}) \epsilon_\ell,\nonumber \\ 
    T_j  \dot \epsilon_j   &=\!- \epsilon_j + (X_j\!-\!X'_j)\!\sum_{\ell=1}^N(H_{j,\ell} + K_{j,\ell}) \epsilon_\ell \\
     & \qquad \qquad  \qquad+(X_j\!-\!X'_j)\! \sum_{\ell=1}^N (A_{j,\ell}\!+\!F_{j,\ell}) \xi_\ell, \nonumber
\end{align}
where the matrices $\vec \Lambda, \vec \Gamma, \vec A, \vec C, \vec F, \vec H, \vec K \in \mathbb{R}^{N \times N}$ (written in component form above) are given by
\begin{equation}
\begin{aligned}\label{eq:def_matrices}
    \Lambda_{j,\ell} &= 
    \left\{ \begin{array}{l l}
        - E_j^{\circ} E_\ell^{\circ} B_{j,\ell} \cos(\delta_\ell ^{\circ}- \delta_j^{\circ} ) \; & \mbox{for} \, j \neq \ell \\
        \sum_{k\neq j}  E_j^{\circ} E_k^{\circ} B_{j,k} \cos(\delta_k ^{\circ}- \delta_j^{\circ} ) \; & \mbox{for} \, j = \ell \\
     \end{array} \right.  \\
    \Gamma_{j,\ell} &= 
    \left\{ \begin{array}{l l}
        - E_j^{\circ} E_\ell^{\circ} G_{j,\ell} \sin(\delta_\ell^{\circ} - \delta_j^{\circ}) \; & \mbox{for} \, j \neq \ell \\
        \sum_{k\neq j}  E_j^{\circ} E_k^{\circ} G_{j,k} \sin(\delta_k^{\circ} - \delta_j^{\circ}) \; & \mbox{for} \, j = \ell \\
     \end{array} \right. \\
     A_{j,\ell} &= 
    \left\{ \begin{array}{l l}
        - E_\ell^{\circ} B_{j,\ell} \sin(\delta_\ell ^{\circ}- \delta_j^{\circ} ) \; & \hspace*{8mm} \mbox{for} \, j \neq \ell \\
        \sum_k E_k^{\circ} B_{j,k} \sin(\delta_k ^{\circ}- \delta_j^{\circ} ) \; & \hspace*{8mm} \mbox{for} \, j = \ell \\
     \end{array} \right.     \\
     C_{j,\ell} &= 
    \left\{ \begin{array}{l l}
        E_j^{\circ} G_{j,\ell} \cos(\delta_\ell^{\circ} - \delta_j^{\circ}) \; & \hspace*{7.5mm} \mbox{for} \, j \neq \ell \\
        \sum_k E_k^{\circ} G_{j,k} \cos(\delta_k^{\circ} - \delta_j^{\circ}) \; & \hspace*{7.5mm} \mbox{for} \, j = \ell \\
     \end{array} \right.  \\
     F_{j,\ell} &= 
    \left\{ \begin{array}{l l}
        -E_\ell^{\circ} G_{j,\ell} \cos(\delta_\ell^{\circ} - \delta_j^{\circ}) \; & \hspace*{4mm} \mbox{for} \, j \neq \ell \\
        \sum_{k\neq j} E_k^{\circ} G_{j,k} \cos(\delta_k^{\circ} - \delta_j^{\circ}) \; & \hspace*{4mm} \mbox{for} \, j = \ell \\
     \end{array} \right.  \\
    H_{j,\ell} &=  B_{j,\ell} \cos(\delta_\ell ^{\circ}- \delta_j^{\circ} ), \\
    K_{j,\ell} &= - G_{j,\ell} \sin(\delta_\ell ^{\circ}- \delta_j^{\circ} ). \\
\end{aligned}
\end{equation}
The diagonal matrices $\vec M$, $\vec D$, $\vec X$, and $\vec T$ (all in $\mathbb{R}^{N \times N}$) comprise the elements $M_j$, $D_j$, $(X_j-X'_j)$, and $T_j$ for $j=1,\ldots,N$, respectively.
All these diagonal matrices are positive definite.

The linearised system~\eqref{eq:linstab1} takes a compact matrix formulation, where the linearised terms are elegantly combined into the Jacobian matrix $\vec{J}\in\mathbb{R}^{3N\times3N}$, by defining the vectors $\vec \xi = (\xi_1,\ldots,\xi_N)^\top$, 
$\vec \nu = (\nu_1,\ldots,\nu_N)^\top$, and $\vec \epsilon = (\epsilon_1,\ldots,\epsilon_N)^\top$, each in $\mathbb{R}^N$, with the superscript $\cdot^\top$ denoting the transpose of a matrix or vector. 
The linearised equations can be written as
\begin{equation}
\frac{\mathrm{d}}{\mathrm{d} t}\!\begin{pmatrix}
\vec{\xi}\\ \vec{\nu} \\ \vec{\epsilon}
\end{pmatrix}=\vec{J}\begin{pmatrix}
\vec{\xi}\\ \vec{\nu} \\ \vec{\epsilon} 
\end{pmatrix},
\end{equation}
with 
\begin{equation}\label{eq:jacobian}
\vec{J} \!=\!\!\begin{pmatrix} \vec 0 & \eye & \vec 0 \\
- \vec M^{-1} ( \vec \Lambda \!+\! \vec \Gamma) &\! - \vec M ^{-1}\! \vec D & \vec{M}^{-1}( \vec A^\top\!\!\! +\! \vec C) \\
\vec T^{-1} \!\vec  X (\vec A\! +\! \vec F) & \vec 0 &  \vec{T}^{-1} \left(\vec X (\vec H\!+\!\vec K) \!-\! \eye\!\!\right)
\end{pmatrix}\!.             
\end{equation}
The Jacobian $\vec J$ can be brought to a different form that clearly portrays the interplay between the matrices comprising the susceptance $B_{j,\ell}$ and the conductance terms $G_{j,\ell}$ of the power lines and machines,
\begin{equation}\label{eq:jacobian_full}
\begin{aligned}
\vec J \!=& \!
    \begin{pmatrix}  \eye &  \vec 0 & \vec 0 \\
    \vec 0 & \vec M^{-1}  & \vec 0 \\
    \vec 0 & \vec 0 &  \vec{T}^{-1} \vec  X \\
    \end{pmatrix}\times\\
    &\quad\left[\!
    \begin{pmatrix} \vec 0 & \eye & \vec 0 \\
    - \vec \Lambda & - \vec D & \vec A^\top  \\
    \vec A  & \vec 0 &  \vec H - \vec  X^{-1} \\
    \end{pmatrix}  
    \! + \!
    \begin{pmatrix} \vec 0  & \vec 0 & \vec 0 \\
    -\vec \Gamma & \vec 0 & \vec C \\
    \vec F & \vec 0 &  \vec K \\
    \end{pmatrix}\! \right]\!.
\end{aligned}
\end{equation}
This decomposition is conspicuously designed to work out the impact of Ohmic losses.
The left matrix in the brackets includes all terms that are present in a lossless grid, and the right matrix composed of the block matrices $\vec \Gamma, \vec C, \vec F, \vec K$ embodies all the matrices associated with resistive losses.
The cleavage into two parts will prove useful hence onward.

\subsection{Linear Stability and Eigenvalues of the Jacobian}\vspace*{-2.5mm}
An equilibrium $(\delta^\circ_j, \omega^\circ_j, E^\circ_j)$ is linearly stable if perturbations in the linearised system~\eqref{eq:linstab1} decay exponentially.
In general, this is the case if and only if all eigenvalues of the Jacobian matrix $\vec J$ have a negative real part~\cite{Kuznetsov2004,Strogatz2015}. 

In the present case one has to take into account that the dynamical system incorporates a fundamental symmetry.
\begin{equation}
\begin{aligned}
\Psi_\alpha: &~ \vec{\delta} \mapsto \vec{\delta} + \alpha \, \vec 1,\\
&~\mathbb{S}^N \to \mathbb{S}^N,
\end{aligned}
\end{equation}
where $\vec 1$ is a vector of ones and $\alpha \in \mathbb{R}$. A shift of all nodal phase angles by a constant value does not have any physical effects: all flows, currents and stability properties remain unaffected.
A geometric interpretation of this symmetry is obtained by viewing the desired operation of the power grid as a limit cycle.
As all points along the cycle are equivalent for power grid operation, one can take an arbitrary point as a representative of the limit cycle and refer to it as `the equilibrium'. 

As a consequence of this symmetry, any perturbation corresponding to a global phase shift or a shift along the limit cycle, respectively, should be excluded from the stability analysis. 
This allows reducing the analysis to the perpendicular subspaces of this symmetry, which are defined as 
\begin{equation}
\begin{aligned}
   \mathcal{D}^{(3)}_\perp &= \left\{ (\vec \xi,\vec{\nu}, \vec \epsilon)\in \mathbb{R}^{3 N}|\vec{1}^\top\vec{\xi}=0\right\},\\
   \mathcal{D}^{(2)}_\perp &= \left\{ (\vec \xi, \vec \epsilon) \in \mathbb{R}^{2N}|\vec{1}^\top\vec{\xi}=0\right\},\\
   \mathcal{D}^{(1)}_\perp &= \left\{ \vec \xi \in \mathbb{R}^{N}|\vec{1}^\top\vec{\xi}=0\right\}.
  \end{aligned}
\end{equation}
These subspaces are always one dimension smaller than the over-branching space.
The subscript $\mathcal{D}^{(\cdot)}_\perp$ refers to the orthogonality devised here, i.e., these spaces are orthogonal to the stable limit-cycle manifold.

Having defined the spaces of operation, one turns to the Jacobian matrix~\eqref{eq:jacobian_full} to unravel the definition of linear stability.
Consider the eigenvalues $\mu_1,\mu_2,\dots,\mu_{3N}\in\mathbb{C}^{3N}$ of the Jacobian defined via
\begin{equation}\label{eq:eigenval1}
 \matr{J} \begin{pmatrix} {\vec{\xi}} \\ {\vec{\nu}} \\ {\vec{\epsilon}}\end{pmatrix} = \mu\begin{pmatrix} {\vec{\xi}} \\ {\vec{\nu}} \\ {\vec{\epsilon}}\end{pmatrix} .
\end{equation} 
There is always one vanishing eigenvalue $\mu_1 = 0$ corresponding to the global shift of all nodal phases, as discussed above.
One excludes this mode from the definition of stability and orders the remaining eigenvalues according to their real parts, without loss of generality,
\begin{equation}
    \mu_1 = 0, \quad
    \Re({\mu_2}) \le \Re({\mu_3}) \le \cdots \le 
    \Re({\mu_{3N}}).
\end{equation}

\subsection{Alternative formulations of the eigenvalue problem}\vspace*{-2.5mm}
We note that the eigenvalue problem~\eqref{eq:eigenval1} can be reformulated in different ways, which are useful for both analytic studies and numerical computation. 
First, one can obtain the eigenvalues $\mu$ from a generalised eigenvalue problem,
\begin{equation}
\begin{aligned}\label{eq:eigen-gep}
    & \begin{pmatrix}
     -\matr \Lambda - \matr \Gamma & \matr  0 & \matr A^\top+ \matr C \\
     \matr A + \matr F & \matr 0 & \matr H - \matr X^{-1} + \matr K   \\
     \matr 0 &  \matr M  & \matr 0
    \end{pmatrix} 
    \begin{pmatrix}
        \vec \xi  \\ \vec \nu \\ \vec \epsilon
    \end{pmatrix}\\  
    & \quad = \mu
    \begin{pmatrix}
     \matr D & \matr M & \matr 0\\
     \matr 0 & \matr 0 & \matr X^{-1} \matr T\\
     \matr M & \matr 0 & \matr 0 
    \end{pmatrix} 
    \begin{pmatrix}
        \vec \xi  \\ \vec \nu \\ \vec \epsilon
    \end{pmatrix}.
\end{aligned}
\end{equation}
To see this, we decompose the original problem~\eqref{eq:eigenval1} in components
\begin{subequations}\label{eq:eigenvalue-comp}
    \begin{align}
    \vec \nu &= \mu \vec \xi , \label{eq:eigenvalue-comp_a}\\
    - \matr M^{-1} \left( (\matr \Lambda+\matr \Gamma) \vec \xi + \matr D \vec \nu - ({\matr A}^\top + \matr C) \vec \epsilon \right)  &= \mu \vec \nu , \quad \quad \label{eq:eigenvalue-comp_b}\\
    \matr T^{-1} \matr X \left( (\matr A+\matr F) \vec \xi +  ({\matr{H}}- \matr X^{-1} + \matr K) \vec \epsilon \right) &= \mu \vec \epsilon.  \label{eq:eigenvalue-comp_c}
    \end{align}
\end{subequations}
We multiply~\eqref{eq:eigenvalue-comp_a} with $\matr M$ and~\eqref{eq:eigenvalue-comp_c} with $\matr X^{-1} \matr T$.
Furthermore, we substitute~\eqref{eq:eigenvalue-comp_a} in~\eqref{eq:eigenvalue-comp_b} and multiply the resulting equation with $\matr M$ and obtain
\begin{equation}
\begin{aligned}
    \matr M \vec \nu &= \mu \matr M \vec \xi, \\
    - (\matr \Lambda+\matr \Gamma) \vec \xi + ({\matr A}^\top+\matr C) \vec \epsilon &=
    \mu \left( \matr M \vec \nu + \matr D  \xi \right), \\
    (\matr A+\matr F) \vec \xi +  (\matr H-\matr X^{-1}\!+\matr K) \vec \epsilon
    &= \mu \matr X^{-1} \matr T \vec \epsilon.
\end{aligned}    
\end{equation}
In matrix form this leads to~\eqref{eq:eigen-gep}.

Second, one can obtain the eigenvalue $\mu$ from a nonlinear eigenvalue problem in a lower dimensional space
\begin{equation}
\begin{aligned}
    & \Bigg[ 
    \begin{pmatrix}
       - \matr \Lambda - \matr \Gamma & \matr{A}^\top + \matr C \\
       \matr A + \matr F & {\matr{H}} - \matr X^{-1}\! + \matr K
    \end{pmatrix}
    - \mu 
    \begin{pmatrix}
       \matr D & \matr 0 \\
       \matr 0 & \matr T^{-1} \matr X
    \end{pmatrix}  \\ 
    & \qquad \qquad
    - \mu^2 
    \begin{pmatrix}
       \matr M & \matr 0 \\
       \matr 0 & \matr 0
    \end{pmatrix}
    \Bigg] 
    \begin{pmatrix}
        \vec \xi \\ \vec \epsilon 
    \end{pmatrix}
     = \begin{pmatrix}
        \vec 0 \\ \vec 0
    \end{pmatrix}.
    \label{eq:eigen-nonlinear}
\end{aligned}
\end{equation}
The remaining component of the eigenvector is then fixed as $\vec \nu = \mu \vec \xi$. We derive this reformulation starting again from the decomposition~\eqref{eq:eigenvalue-comp}. 
Substituting~\eqref{eq:eigenvalue-comp_a} into~\eqref{eq:eigenvalue-comp_b} and multiplying with $\matr M$ eliminates $\vec \nu$. 
Furthermore, we multiply~\eqref{eq:eigenvalue-comp_c} with $\matr T^{-1} \matr X$ and obtain
\begin{equation}
\begin{aligned}
     - (\matr \Lambda+\matr \Gamma) \vec \xi + ({\matr A}^\top+\matr C) \vec \epsilon
     -\mu \matr D \vec \xi - \mu ^2 \matr M \vec \xi &=0, \\
    (\matr A+\matr F) \vec \xi + (\matr H - \matr X^{-1}\! + \matr K) \vec \epsilon
    - \mu \matr T^{-1} \matr X \vec \epsilon 
    &= 0.
\end{aligned}
\end{equation}
In matrix form this results in~\eqref{eq:eigen-nonlinear}.
With this in hand, we now introduce the main lemmata of this work that shall pave the way to several analytical criteria in the later sections.

\section{Analytic stability results}\label{sec:IV}\vspace*{-2.5mm}
\subsection{The lossless case}\vspace*{-2.5mm}
The lossless case was previously analysed in detail in Ref.~ \cite{Sharafutdinov2018}, so we only review the essential results very briefly. 
Most importantly, linear stability is determined by a reduced, hermitian Jacobian matrix. 
We state this result in the following lemma.
\begin{lemma}\label{lemma:Xi}
The linear stability of an equilibrium $(\delta^\circ_j, \omega^\circ_j, E^\circ_j)$ is determined by the reduced Jacobian
\begin{equation}
   \matr \Xi = 
    \begin{pmatrix}
    - \matr \Lambda & \matr A^\top \\
    \matr A &  \matr H - \matr X^{-1} 
    \end{pmatrix}.
    \label{eq:def-Xi-lossless}
\end{equation}
The equilibrium is stable if $\matr \Xi$ is negative definite on ${\mathcal{D}}_\perp^{(2)}$. It is unstable if $\matr \Xi$ is not negative semi-definite.
\end{lemma}
\begin{proof}
Define the Lyapunov function candidate
\begin{equation}
    V = \begin{pmatrix} 
    \vec{\nu} \\
    \vec{\xi} \\
    \vec{\epsilon}\\
    \end{pmatrix}^{\!\!\top}\!\!
    \underbrace{
    \begin{pmatrix} \vec M & \vec 0 & \vec 0 \\
     \vec 0 & \vec \Lambda & -\vec A^\top  \\
     \vec 0 & -\vec A  &  -\left(\vec H\! - \!\vec  X^{-1}\right) \\
    \end{pmatrix}
    }_{=: \matr P}
    \begin{pmatrix} 
    \vec{\nu} \\
    \vec{\xi} \\
    \vec{\epsilon}\\
    \end{pmatrix}
\end{equation}
Then one finds
\begin{equation}
\begin{aligned}
  \dot{V} &= - \dot{\vec{\nu}}^\top\!\vec{M}\vec{\nu} - \vec{\nu}^\top\!\vec{M}\dot{\vec{\nu}}  + \dot{\vec{\xi}}^\top\!\vec{\Lambda}\vec{\xi} + \vec{\xi}^\top\!\vec{\Lambda}\dot{\vec{\xi}}  - \dot{\vec{\epsilon}}^\top\!\vec{A}\vec{\xi} - \vec{\epsilon}^\top\!\vec{A}\dot{\vec{\xi}}  \\ 
  & - \dot{\vec{\xi}}^\top\!\vec{A}^\top\vec{\epsilon} - \vec{\xi}^\top\!\vec{A}^\top\dot{\vec{\epsilon}}- \dot{\vec{\epsilon}}^\top\!\left(\vec{H}\!-\! \vec{X}^{-1}\right)\vec{\epsilon} - \vec{\epsilon}^\top\!\left(\vec{H}\!-\! \vec{X}^{-1}\right)\dot{\vec{\epsilon}}\\
&= - 2{\vec{\nu}}^\top\!\vec{D}\vec{\nu} -   2\!\left[\vec{\xi}^\top\!\vec{A}^\top \vec{X} \vec{T}^{-1}\!\vec{A}\dot{\vec{\xi}} \right. \\
& \quad +\vec{\epsilon}^\top\!\left(\vec{H}\!-\! \vec{X}^{-1}\right) \vec{X} \vec{T}^{-1}\!\vec{A}\dot{\vec{\xi}} \\
& \quad + \vec{\xi}^\top\!\vec{A}^\top\!\vec{T}^{-1}\!\vec{X} \left(\vec{H}\!-\! \vec{X}^{-1}\right)\dot{\vec{\epsilon}} \\
& \quad + \vec{\epsilon}^\top\!\left(\vec{H}\!-\! \vec{X}^{-1}\right) \vec{T}^{-1}\!\vec{X} \left(\vec{H}\!-\! \vec{X}^{-1}\right)\dot{\vec{\epsilon}} \big]\\
&= - 2{\vec{\nu}}^\top\!\vec{D}\vec{\nu}  \\
& - 2\!\left[\vec{\xi}^\top\!\vec{A}^\top\! + \vec{\epsilon}^\top\!\left(\vec{H}\!-\! \vec{X}^{-1}\right)\right]\!\vec{X} \vec{T}^{-1}\!\left[\vec{A}\vec{\xi} + \left(\vec{H}\!-\! \vec{X}^{-1}\right)\vec{\epsilon}\right] \\
&<0.
\end{aligned}
\end{equation}
The last inequality follows as the matrices $\matr X, \matr T, \matr D$ are diagonal with only positive entries.
If $\matr \Xi$ is negative definite, then $\matr P$ is positive definite and the equilibrium is stable according to Lyapunov's stability theorem. 
If $\matr \Xi$ is not negative semi-definite, then also $\matr P$ is not positive semi-definite and the equilibrium is unstable according to Lyapunov's instability theorem.
\end{proof}

The reduced Jacobian can be further decomposed into the subspace corresponding to perturbations of the angles or voltages, respectively. 
This is especially helpful for the derivation of rigorous stability criteria, cf.~\cite{Sharafutdinov2018}. 

\begin{lemma}\label{lemma:Schur-lossless}(Sufficient and necessary stability conditions for lossless systems).
\begin{itemize}\itemsep0em 
\item[I.] The equilibrium $(\delta^\circ_j, \omega^\circ_j, E^\circ_j)$ of the lossless grid is linearly stable if 
(a) the matrix $\vec \Lambda$ is positive definite on $\mathcal{D}^{(1)}_\perp$ and
(b) the matrix $\vec H - \vec X^{-1} + \vec A \vec \Lambda^+ \vec A^\top$ is negative definite,
where $\cdot^{+}$ is the Moore--Penrose pseudoinverse.
The equilibrium is unstable if any of the two matrices is not negative semi-definite. 
\item[II.] The equilibrium $(\delta^\circ_j, \omega^\circ_j, E^\circ_j)$  of the lossless grid is linearly stable if
(a) the matrix $\vec H - \vec X^{-1}$ is negative definite and
(b) the matrix $\vec \Lambda +  \vec A^\top (\vec H - \vec X^{-1})^{-1} \vec A$ is positive definite on $\mathcal{D}^{(1)}_\perp$. The equilibrium is unstable if any of the two matrices is not negative semi-definite. 
\end{itemize}
\end{lemma}

This result follows from Lemma~\ref{lemma:Xi} by applying the Schur complement, where some care has to be taken to distinguish definiteness and semi-definiteness as well as about the domain of the matrices~\cite{Zhang2006}. 
Details are given in Sharafutdinov \textit{et al.}~\cite{Sharafutdinov2018}.

\subsection{The lossy case}\vspace*{-2.5mm}
We now drop the simplification of a lossless grid and analyse how the presence of resistive terms alters the linear stability of the grid. Complete rigorous results are hard to obtain as the relevant matrices are non longer Hermitian. 
In particular, the one-to-one correspondence between the definiteness and the signs of the eigenvalues does no longer apply. 
However, we show that the above results can be generalised in a straightforward way to leading order in the losses.

The central objective of interest is again the reduced Jacobian $\vec \Xi$, which now reads 
\begin{equation}
  \matr \Xi = \begin{pmatrix}
      - \matr \Lambda - \matr \Gamma & \matr A^\top + \matr C\\
      \matr A + \matr F &  \matr H - \matr X^{-1} + \matr K 
      \end{pmatrix}.
\end{equation}
For the further analysis, we decompose it into its hermitian and anti-hermitian part, $\matr \Xi = \matr \Xi_H + \matr \Xi_A$, with
\begin{equation}\label{eq:def-Xi-lossy}
\begin{aligned}
    \matr \Xi_H &= \frac{1}{2}\left( \matr \Xi + \matr \Xi^\top \right)\\
    &= \begin{pmatrix}
      - \matr \Lambda - \matr \Gamma^d & \matr A^\top + \matr N \\
      \matr A + \matr N &  \matr H - \matr X^{-1} 
      \end{pmatrix} \\
    \matr \Xi_A &= \frac{1}{2}\left( \matr \Xi - \matr \Xi^\top\right).
\end{aligned}
\end{equation}
with the hermitian matrices
\begin{equation}\label{eq:def_of_gamma_d}
\begin{aligned}
    \matr \Gamma^d &= \frac{1}{2} \left( \matr \Gamma + \matr \Gamma^\top \right)  \\  
    \matr N &= \frac{1}{2} \left( \matr C + \matr F \right).
\end{aligned}
\end{equation}
One finds that these matrices are all \emph{diagonal} with entries
\begin{equation}
\begin{aligned}
    \Gamma_{j,j}^\mathrm{d} &= \sum_{k\neq j}^N  E_j^{\circ} E_k^{\circ} G_{j,k} \sin(\delta_k^{\circ} - \delta_j^{\circ}),\\
    N_{j,j} &= \sum_{k\neq j}^N E_k^{\circ} G_{j,k} \cos(\delta_k^{\circ} - \delta_j^{\circ}).
\end{aligned}
\end{equation}

We start with providing a rigorous sufficient stability condition, that generalises the condition of negative definiteness of the reduced Jacobian in Lemma~\ref{lemma:Xi}. 

\begin{lemma}
\label{lemma:sufficient-lossy}
The lossy microgrid is stable if for all vectors $\vec x \in \mathbb{C}^{2N}$ 
\begin{align}
    \vec x^\dagger  \matr \Xi_H  \vec x
    < - \frac{\left( \vec x^\dagger  \matr \Psi   \vec x \right)}{
    \left(\vec x^\dagger  \matr \Phi  \vec x \right)^2
    }
    \left( \vec x^\dagger  \matr \Xi_A \vec x \right)^2
    \label{eq:suff-lossy}
\end{align}
with the abbreviations
\begin{equation}
    \Phi = \begin{pmatrix} 
          \matr D  & \matr{0}  \\
          \matr{0} & \matr T^{-1} \matr X  \\
          \end{pmatrix}, \quad
    \Psi = \begin{pmatrix} 
          \matr M  & \matr{0}  \\ \matr{0}  & \vec{0}  \\
          \end{pmatrix}.
\end{equation}
\end{lemma}
\begin{proof}
We start from the nonlinear eigenvalue problem~\eqref{eq:eigen-nonlinear} and multiply from the left with the hermitian conjugate of the eigenstate
$(\vec \xi^\dagger, \vec \epsilon^\dagger )$ to obtain the algebraic equation
\begin{equation}
  \eta_1 + \eta_2 \mu + \eta_3 \mu^2 = 0
  \label{eq:eigenval-quadeq1}
\end{equation}
with 
\begin{equation}
\begin{aligned}
     \eta_1 &= - \begin{pmatrix} \vec \xi \\ \vec \epsilon \end{pmatrix}^{\dagger} 
        \matr \Xi
     \begin{pmatrix} \vec \xi \\ \vec \epsilon \end{pmatrix}, \\
     \eta_2 &= \begin{pmatrix} \vec \xi \\ \vec \epsilon \end{pmatrix}^{\dagger}
              \matr \Phi
              \begin{pmatrix} \vec \xi \\ \vec \epsilon \end{pmatrix} > 0\\
     \eta_3 &= \begin{pmatrix} \vec \xi \\ \vec \epsilon \end{pmatrix}^{\dagger}
              \matr \Psi 
              \begin{pmatrix} \vec \xi \\ \vec \epsilon \end{pmatrix} > 0. 
\end{aligned}
\end{equation}
The coefficients $\eta_2$ and $\eta_3$ are real and strictly positive, as the matrices $\matr D, \matr M, \matr T, \matr X$ are diagonal with strictly positive entries (except for the trivial case $\vec \xi = \vec 0$ for which $\eta_3=0$). 
For the remaining coefficient $\eta_1$ we write
\begin{equation}
\begin{aligned}
    \eta_1 &= \alpha + i \beta, \\
    \Rightarrow \; \alpha &= - \begin{pmatrix} \vec \xi \\ \vec \epsilon \end{pmatrix}^{\dagger} 
        \matr \Xi_H
     \begin{pmatrix} \vec \xi \\ \vec \epsilon \end{pmatrix}, \\
     \beta &= - \begin{pmatrix} \vec \xi \\ \vec \epsilon \end{pmatrix}^{\dagger} 
        \matr \Xi_A
     \begin{pmatrix} \vec \xi \\ \vec \epsilon \end{pmatrix}.
\end{aligned}
\end{equation}
The algebraic equation~\eqref{eq:eigenval-quadeq1} can now be solved for $\mu$ such that
\begin{equation}
    \mu = \frac{-\eta_2 \pm \sqrt{\eta_2^2 - 4 \eta_3 (\alpha+i\beta)}}{2 \eta_3}.
\end{equation}
To ensure stability, the real part of $\mu$ needs to be strictly smaller than zero, which translates to
\begin{equation}
    \Re\left( \sqrt{\eta_2^2 - 4 \eta_3 (\alpha+i\beta)}   \right) < \eta_2.
\end{equation}
One can now show by an explicit calculation that this is the case if
\begin{equation}
      \alpha > \frac{\eta_3 \beta^2}{\eta_2^2}.
\end{equation}
Hence if the assumption~\eqref{eq:suff-lossy} is satisfied for all vectors, we indeed have $\Re (\mu) < 0$ and the equilibrium is stable.
\end{proof}

Lemma~\ref{lemma:sufficient-lossy} provides a rigorous sufficient condition for linear stability in lossy grids. 
However, it might be hard to apply in practice due to its nonlinearity. 
Nevertheless, we can draw some important general conclusions. 
First, the condition~\eqref{eq:suff-lossy} is stricter than in the lossless case as the right-hand side~\eqref{eq:suff-lossy} is generally smaller than zero.
However, this right-hand side is of quadratic order in the losses.
Hence -- to leading order in the losses -- negative definiteness of $\matr \Xi_H$ still guarantees stability. 

We now extend this argument. 
We show that -- to leading order in the losses -- the eigenvalues $\mu_n$ that encode the linear stability are determined only by the hermitian matrix $\matr \Xi_H$. 
We make this statement precise for the generalised eigenvalue problem formulated in~\eqref{eq:eigen-gep}. 

\begin{lemma}\label{lemma:4}
The eigenvalues of the Jacobian are given by the hermitian generalised eigenvalue problem
\begin{equation}
\begin{aligned}
    & \begin{pmatrix}
     -\matr \Lambda - \matr \Gamma^d & \matr A^\top + \matr N & \matr  0 \\
     \matr A + \matr N & \matr H - \matr X^{-1} \matr E &  \matr 0 \\
     \matr 0 & \matr 0 & \matr M
    \end{pmatrix} \vec x_n \\ 
    & \qquad = \mu_n
    \begin{pmatrix}
     \matr D & \matr 0 &  \matr M \\
     \matr 0 & \matr X^{-1} \matr T & \matr 0 \\
     \matr M & \matr 0 & \matr 0 
    \end{pmatrix} \vec x_n.
    \label{eq:gep-lossy-leading}
\end{aligned}
\end{equation}
up to corrections of quadratic order in $\matr \Xi_A$. 
\end{lemma}

\begin{proof}
The result is proven using a standard perturbation theory argument, treating $\matr \Xi_A$ as a small perturbation. 
For the sake of convenience, we abbreviate the matrices in the generalised eigenvalue problem~\eqref{eq:eigen-gep} such that we have the equation
\begin{equation}
    \matr{\mathcal{A}} \vec x_n = \mu_n \matr{\mathcal{B}} \vec x_n.
\end{equation}
We write
\begin{equation}
    \matr{\mathcal{A}} = \matr{\mathcal{A}}_H + \varepsilon \matr{\mathcal{A}}_A,
\end{equation}
where $\matr{\mathcal{A}}_H$ is the hermitian part and $\matr{\mathcal{A}}_A$ is the anti-hermitian part treated as a perturbation. 
We now 
expand eigenstates and eigenvalues as
\begin{equation}
\begin{aligned}\label{eq:perturb}
    \mu_n &= \mu_n^{(0)} + \varepsilon \mu_n^{(1)} +  \varepsilon^2 \mu_n^{(2)} + \ldots, \\ 
    \vec x_n &= \vec x_n^{(0)} + \varepsilon \vec x_n^{(1)} + \varepsilon^2 \vec x_n^{(2)} + \ldots,
\end{aligned}
\end{equation}
and substitute this Ansatz into the generalised eigenvalue problem. 
To zeroth order in $\varepsilon$ we obtain 
\begin{equation}
    \matr{\mathcal{A}}_H \vec x_n^{(0)} = \mu_n^{(0)} \matr{\mathcal{B}} \vec x_n^{(0)},
\end{equation}
that is, we obtain~\eqref{eq:gep-lossy-leading}. 
We note that this problem is hermitian, such that unperturbed eigenvectors can be chosen as real and normalised as
\begin{equation}
    \vec x_m^{(0) \top} \matr{\mathcal{B}} \vec x_n^{(0)} = \delta_{m,n}.
\end{equation}
To first order in $\varepsilon$ we obtain
\begin{equation}
    \matr{\mathcal{A}}_A \vec x_n^{(0)} + \matr{\mathcal{A}}_H \vec x_n^{(1)}
    = \mu_n^{(0)} \matr{\mathcal{B}} \vec x_n^{(1)} 
    + \mu_n^{(1)} \matr{\mathcal{B}} \vec x_n^{(0)}.
\end{equation}
Multiplying from the left by $\vec x_n^{(0)\top}$ and exploiting that $\vec x_n^{(0)\top} \matr{\mathcal{A}}_H = \mu_n^{(0)} \vec x_n^{(0)\top} \matr{\mathcal{B}}$ yields
\begin{equation}\label{eq:first_order_mu}
    \mu_n^{(1)} = \frac{\vec x_n^{(0) \top} \matr{\mathcal{A}}_A \vec x_n^{(0)} }{\vec x_n^{(0) \top} \matr{\mathcal{B}} \vec x_n^{(0)} } \, .
\end{equation}
Now we can use the fact that $\matr{\mathcal{A}}_A$ is anti-symmetric to obtain
\begin{equation}
\begin{aligned}
    &\vec x_n^{(0) \top} \matr{\mathcal{A}}_A \vec x_n^{(0)}
    = \left(  \vec x_n^{(0) \top} \matr{\mathcal{A}}_A \vec x_n^{(0)}    \right)^\top \\
    &\quad = \vec x_n^{(0) \top} \matr{\mathcal{A}}_A^\top \vec x_n^{(0)}
    = - \vec x_n^{(0) \top} \matr{\mathcal{A}}_A \vec x_n^{(0)}.
\end{aligned}
\end{equation}
Hence, we have
\begin{equation}
    \vec x_n^{(0) \top} \matr{\mathcal{A}}_A \vec x_n^{(0)} = 0
    \quad \Rightarrow \quad
    \mu_n^{(1)} = 0.
\end{equation}
That is, the linear order correction to the eigenvalues vanishes, leaving terms of quadratic or higher order. 
\end{proof}

We conclude that -- to leading order in the losses -- only the hermitian part of the Jacobian is relevant for stability.
We can generalise all results from the lossless case if we replace the reduced Jacobian~\eqref{eq:def-Xi-lossless} by the matrix $\matr \Xi_H$ defined in~\eqref{eq:def-Xi-lossy}. 
In particular, Lemma~\ref{lemma:Schur-lossless} is generalised as follows. 

\begin{lemma}
\label{lemma:Schur-lossy}\
To leading order in the Ohmic losses the linear stability of an equilibrium  $(\delta_j^{\circ},\omega_j^{\circ},E_j^{\circ})$ is determined by the hermitian part of the reduced Jacobian matrix:
stable if  $\matr \Xi_H$ is negative definite on $\mathcal{D}_\perp^{(2)}$ and unstable if $\matr \Xi_H$  is not negative semi-definite. 
Stability conditions for this matrix can be decomposed as follows:
\begin{itemize}\itemsep0em 
\item[I.] The matrix $\matr \Xi_H$ is negative definite on $\mathcal{D}_\perp^{(2)}$ if 
(a) the matrix $\vec \Lambda + \matr \Gamma^d$ is positive definite on $\mathcal{D}^{(1)}_\perp$ and
(b) the matrix $\vec H - \vec X^{-1} + (\matr A + \matr N)  (\vec \Lambda + \matr \Gamma^d)^+ (\matr A+ \matr N)^\top$ is negative definite.
The matrix $\matr \Xi_H$ is non-negative semi-definite if any of the two matrices is not negative semi-definite. 
\item[II.] The matrix $\matr \Xi_H$ is negative definite on $\mathcal{D}_\perp^{(2)}$ if
(a) the matrix $\vec H - \vec X^{-1}$ is negative definite and
(b) the matrix $(\vec \Lambda + \matr \Gamma^d)  + (\matr A + \matr N)^\top (\vec H - \vec X^{-1})^{-1} \vec (\matr A + \matr N)$ is positive definite on $\mathcal{D}^{(1)}_\perp$.
The matrix  $\matr \Xi_H$ is non-negative semi-definite if any of the two matrices is not negative semi-definite.
\end{itemize}
\end{lemma}

We will henceforth work with Lemma~\ref{lemma:Schur-lossy}, where we note that the lossless case is recovered when $\matr \Gamma^d=\matr N=\matr 0$ and we return to Lemma~\ref{lemma:Schur-lossless}.

\section{Explicit stability criteria}\label{sec:V}\vspace*{-2.5mm}

\subsection{Angle vs. voltage stability}\vspace*{-2.5mm}

The decomposition of the reduced Jacobian in Lemma~\ref{lemma:Schur-lossy} is of fundamental importance to this work, as it evinces the roles of the rotor-angle and the voltage dynamics for the stability of the third-order model.

Consider first the isolated power-angle dynamics, assuming that the voltages $E_j$ remain fixed. 
Fixing $\vec \epsilon = 0$, the linearised equations of motions read
\begin{equation}
  \frac{\mathrm{d}}{\mathrm{d} t}
  \begin{pmatrix}
    \vec{\xi}\\ \vec{\nu}
  \end{pmatrix} =
  \begin{pmatrix} \vec 0 & \eye \\
      - \vec M^{-1} ( \vec \Lambda + \vec \Gamma) 
      & - \vec M ^{-1} \vec D 
   \end{pmatrix}
  \begin{pmatrix}
      \vec{\xi}\\ \vec{\nu}
   \end{pmatrix}.
\end{equation}
Performing the same simplification as in the previous section, one finds that the isolated rotor-angle dynamics is linearly stable -- to leading order in the losses -- if and only if the matrix $\vec \Lambda + \vec \Gamma^\mathrm{d}$ is positive definite on $\mathcal{D}^{(1)}_\perp$. 

Similarly, consider the isolated voltage dynamics by assuming that the rotor angle remains fixed. 
Fixing $\vec \nu = \vec \xi = 0$, the linearised equations of motion read
\begin{equation}
    \frac{\mathrm{d}}{\mathrm{d} t} \vec \epsilon =
     \vec{T}^{-1} \vec X ( \vec H  - 
     \vec X^{-1} + \vec K) \,
    \vec \epsilon.    
\end{equation}
Hence, one finds that the isolated voltage dynamics is linearly stable -- to leading order in the losses -- if and only if the matrix $\vec H - \vec X^{-1}$ is negative definite.

In conclusion, one finds that the criteria I.~\!(a) and II.~\!(a) in Lemma~\ref{lemma:Schur-lossy} ensure the stability, to linear order in the losses, of the isolated rotor-angle or voltage subsystem, respectively. 
Linear stability of the entire system is ensured if and only if, in addition, the complementary criteria I.~\!(b) or II.~\!(b) are satisfied. 

To further elucidate the nature of the stability conditions, consider the full stability criterion I. in Lemma~\ref{lemma:Schur-lossy}.
Assume that criterion I.~\!(a) is satisfied, i.e., $\vec \Lambda + \vec \Gamma^\mathrm{d}$ is positive definite on $\mathcal{D}^{(1)}_\perp$, and the rotor-angle subsystem is linearly stable to leading order in the losses. 
The complementary criterion I.~\!(b) can then be written as
\begin{equation}\label{eq:conIIb-succ}
   \vec H - \vec X^{-1} \prec - (\vec A + \vec N) (\vec \Lambda + \vec \Gamma^\mathrm{d})^{+} (\vec A + \vec N)^\top,
\end{equation}
where $\prec$ denotes negative definiteness (equivalently, $\succ$ positive definiteness).
This condition is far stricter than the condition of pure voltage stability, $\vec H - \vec X^{-1} \prec 0$. 
Hence, stability of the two isolated subsystems is not sufficient, instead they must comprise a certain `security margin' quantified by the right-hand side of~\eqref{eq:conIIb-succ} in order to maintain linear stability.

Making use of the angle-voltage decomposition, one can derive explicit necessary and sufficient stability criteria. 
To this end, we first consider the isolated subsystems and subsequently the composite dynamics of the full system. 
Note that the lossless case has been discussed in Ref.~\cite{Sharafutdinov2018}, thus here the focus is placed on the impact of Ohmic losses in leading order.

\subsection{Voltage stability}\vspace*{-2.5mm}
Criterion II.~\!(a) in Lemma~\ref{lemma:Schur-lossy} entails the stability of the isolated voltage subsystem -- up to leading order in the losses.
A violation implies the instability of the voltage dynamics, and as a consequence also the instability of the entire system, including the rotor-angle and frequency dynamics. 

Most remarkably, criterion II.~\!(a) includes only the matrices $\vec H$ and $\vec X$, which are also present in the lossless case~\cite{Sharafutdinov2018}.
To leading order, ohmic losses in the transmission lines thus affect voltage stability only indirectly via the position of the respective equilibrium, in particular via the equilibrium rotor angles $\delta_j^\circ$, which enter the matrix $\vec H$. 
Due to the similarity to the lossless case, this work refrains from a detailed analysis of voltage stability and only quotes two results from Sharafutdinov \textit{et al.}~\cite{Sharafutdinov2018}.

\begin{cor}\label{cor:1}
If for all nodes $j=1,\ldots,N$
\begin{equation}
    (X_j - X_j')^{-1} > \sum_{\ell=1}^N B_{j,\ell}, 
\end{equation}
then the matrix $\vec H - \vec X^{-1}$ is negative definite.
\end{cor}

\begin{cor}\label{cor:2}
If for any subset of nodes $\mathcal{S} \subset \{1,2,\ldots,N\}$,
\begin{align}
    \sum_{j \in \mathcal{S}} (X_j - X_j')^{-1} \le \sum_{j,\ell \in \mathcal{S}} 
                     B_{j,\ell} \cos(\delta_\ell^\circ - \delta_j^\circ),
\end{align}
then the matrix $\vec H - \vec X^{-1}$ is not negative definite and the necessary stability condition in Lemma~\ref{lemma:Schur-lossy} is violated.
\end{cor}

\subsection{Rotor-angle stability}\vspace*{-2.5mm}

Criterion I.~\!(a) in Lemma~\ref{lemma:Schur-lossy} entails the stability of the isolated rotor-angle subsystem. 
Briefly take the lossless case into consideration, for which rotor-angle stability is determined by the matrix $\vec \Lambda$. 
The isolated subsystem is stable if $\vec \Lambda$ is positive definite on $\mathcal{D}_\perp^{(1)}$ or, equivalently, if the eigenvalues satisfy $0 < \lambda_2 < \cdots < \lambda_N$. 
One can directly derive sufficient stability criteria in terms of the angle differences in the grid: If for all connections $(j,\ell)$ in a power grid one has
\begin{equation}
   \cos\left(\delta_j^\circ - \delta_\ell^\circ \right) > 0, 
\end{equation}
then the isolated rotor-angle subsystem is stable.
This follows from the fact that $\vec \Lambda$ is a proper Laplacian matrix of a weighted undirected graph, which is well known to be positive definite on $\mathcal{D}_\perp^{(1)}$. 
If the condition is not satisfied for a line, the matrix $\vec \Lambda$ rather describes a signed graph, for which positive definiteness is more involved~\cite{Manik2014}. 
Sufficient and necessary criteria have been obtained in Refs.~\cite{Zelazo2014,Song2015,Chen2016a,Chen2016b}. 

One can generalise the above condition to power grids with Ohmic losses in the following way.
\begin{cor}\label{cor:3}
If for all connections $(j,\ell)$ in a power grid, one has
\begin{equation}\label{eq:angle-critline}
    B_{j,\ell} \cos\!\left(\delta_\ell^\circ- \delta_j^\circ\right)
  + G_{j,\ell} \sin\!\left(\delta_\ell^\circ- \delta_j^\circ\right)
  > 0, 
\end{equation}
then the eigenvalue $\lambda_2,\ldots,\lambda_N$ of $\vec \Lambda + \vec \Gamma$ have positive real part and the matrix $\vec \Lambda + \vec \Gamma^\text{d}$ is positive definite on $\mathcal{D}_\perp^{(1)}$, such that the isolated angle subsystem is linearly stable to leading order in the losses. 
\end{cor}
\begin{proof}
The statement can be proved by applying Ger\v{s}gorin's circle theorem~\cite{Gerschgorin1931} to $\vec \Lambda + \vec \Gamma$. Each eigenvalue of this matrix $\lambda_j$ is bound to exist in a disk of radius $R_j = \sum_{\ell \neq j} |\Lambda_{j,\ell} + \Gamma_{j,\ell}|$ around the centre $\Lambda_{j,j} + \Gamma_{j,j}$ such that
\begin{equation}
    | \lambda_j - (\Lambda_{j,j} + \Gamma_{j,j})|
    \le  \sum_{\ell \neq j} | \Lambda_{j,\ell} + \Gamma_{j,\ell}|.
\end{equation}
If condition~\eqref{eq:angle-critline} is satisfied, one can simplify this relation to 
\begin{equation}
\begin{aligned}
    | \lambda_j - (\Lambda_{j,j} + \Gamma_{j,j})|
    \le  & \sum_{\ell \neq j}  \Lambda_{j,\ell} + \Gamma_{j,\ell} \\
    & = (\Lambda_{j,j} + \Gamma_{j,j}),
\end{aligned}
\end{equation}
which directly yields
\begin{equation}
    \Re{(\lambda_j)} \ge 0.
\end{equation}
Now one furthers show that $\lambda_1 = 0$ is the only vanishing eigenvalue of $\vec \Lambda + \vec \Gamma$ such that 
\begin{equation}
    \Re({\lambda_j}) > 0, \qquad j=2,\ldots,N.
\end{equation}
For every non-zero vector $\vec x \in \mathcal{D}_\perp^{(1)}$, we thus have
\begin{equation}
   \vec x^\top (\vec \Lambda + \vec \Gamma^\text{d})\, \vec x
   = \Re \left[ \vec x^\top (\vec \Lambda + \vec \Gamma)\, \vec x  \right]
   > 0,
\end{equation}
and $(\vec \Lambda + \vec \Gamma^\text{d})$ is positive definite on $\mathcal{D}_\perp^{(1)}$.
\end{proof}

We see that even the case of rotor-angle stability becomes much more involved in the lossy case due to the presence of the matrix $\vec \Gamma$. 
This holds especially for the interpretation of results in terms of the network structure.
In the lossless case the stability condition can be rephrased as $\lambda_2>0$, which is particularly convenient as $\lambda_2$ is a measure of the network's algebraic connectivity.
Hence, the stability condition can be interpreted in terms of graph topology and connectivity~\cite{Newman2018}. 
This relation no longer applies in the lossy case.
In particular, $\vec \Lambda + \vec \Gamma$ is a Laplacian, but of a directed signed graph.
Hence, the eigenvalues are not guaranteed to be real.
In contrast, the matrix $\vec \Lambda + \vec \Gamma^d$ is hermitian and thus has real eigenvalues, but it is no longer a Laplacian matrix such that the interpretation of its lowest non-zero eigenvalue as a connectivity does no longer hold. 
However, the relation still holds approximately if we restrict ourselves to the leading order impact of Ohmic losses. 

\begin{lemma}\label{lemma:6}
To leading order in the losses, the eigenvalues of $\vec \Lambda + \vec \Gamma$ and $\vec \Lambda + \vec \Gamma^d$ coincide.
\end{lemma}
\begin{proof}
In identical fashion to the proof of Lemma~\ref{lemma:4}, consider 
\begin{equation}
    \left(\vec \Lambda + \vec \Gamma\right)\vec{x}_n = \lambda_n\vec{x}_n,
\end{equation}
with $\vec{x}_n$ and $\lambda_n$ the respective eigenstates and eigenvalues.
Separate $\vec \Lambda + \vec \Gamma$ into a hermitian and anti-hermitian parts and take the anti-hermitian part as a perturbation.
Consider an expansion of the eigenstates and normalised eigenvalues as in~\eqref{eq:perturb}.
To leading order in the losses
\begin{equation}
    \vec x_n^{(0) \top} \left(\vec \Lambda + \vec \Gamma\right) \vec x_n^{(0)} = \vec x_n^{(0) \top} (\vec \Lambda + \vec \Gamma^{\mathrm{d}})\, \vec x_n^{(0)} = \lambda_n,
\end{equation}
thus the eigenvalues of $\vec \Lambda + \vec \Gamma$ and $\vec \Lambda + \vec \Gamma^{\mathrm{d}}$ coincide.
\end{proof}

In the following we will formulate several stability criteria for the full system applying to the leading order in the losses. 
We frequently use the eigenvalues $\lambda_2$, which is assumed to be real and interpreted as a connectivity, and the associated eigenvector $\vec v_F$ called the Fiedler vector~\cite{Fiedler1973, Fiedler1975, Chung1997, Newman2018}.
We stress that this is not necessarily true, but is appropriate to leading order in the losses as shown above.

\subsection{Mixed instabilities}\vspace*{-2.5mm}

We now turn to the interplay of voltage and angle stability, i.e., further investigating criteria I.~\!(b) and II.~\!(b) in Lemma~\ref{lemma:Schur-lossy}.
Unless stated otherwise, consider an equilibrium such that the criteria I.~\!(a) and II.~\!(a) in Lemma~\ref{lemma:Schur-lossy} are satisfied.
Hence, the isolated subsystems are stable, but the full system can still become unstable.

To begin, consider the case where the voltage dynamics are very stiff, i.e., the case where $(X_j-X'_j)$ are small.
Recall that in the limit $(X_j-X'_j) \rightarrow 0$ the voltage dynamics are trivially stable such that stability is determined solely by the angular subsystem.
One can extend this analysis to the case of small but non-zero $(X_j-X'_j)$ and relate stability to the connectivity of the power grid.
The stability of the isolated rotor-angle subsystem is ensured if (cf.~criterion I.~\!(a) in Lemma~\ref{lemma:Schur-lossy}, or Refs.~\cite{Sharafutdinov2018, Doerfler2010})
\begin{equation}
    \Re{(\lambda_2)} > 0,
\end{equation}
where $\lambda_2$ is the lowest non-zero eigenvalue of the Laplacian $\vec \Lambda + \vec \Gamma$, interpreted as the algebraic connectivity, which is real to leading order in the Ohmic losses.

\begin{cor}\label{cor:4}
To leading order in the Ohmic losses, a necessary condition for the stability of an equilibrium point is given by
\begin{align}\label{eq:stab-lambda-small-x}
    \lambda_2 > &~\! \vec v_F^\dagger \left[ \vec A^\top \vec X \vec A + 2\vec A^\top \vec X \vec N + \vec N \vec X \vec N \right] \vec v_F \\
    & \qquad \qquad\qquad \qquad\qquad \qquad+ \mathcal{O}((X_j - X'_j)^2),
\end{align}
where $\vec v_F$ denotes the Fiedler vector of the Laplacian $\vec \Lambda + \vec \Gamma$ for $(X_j - X'_j) \equiv 0$.
\end{cor}

\begin{proof}
The normalised Fiedler vector, at $(X_j - X'_j) \equiv 0$, is denoted $\vec v_F$. The actual normalised Fiedler vector, for a particular non-zero value of the $(X_j - X'_j)$, is denoted $\vec v'_F$, such that
\begin{equation}
   \vec v'_F = \vec v_F + \mathcal{O}((X_j - X'_j)^1).
\end{equation}
Take the expansion
\begin{equation}
  -(\vec H - \vec X^{-1})^{-1} = (\vec X^{-1}- \vec H)^{-1} =  \sum_{\ell = 0}^\infty \vec X(\vec X \vec H)^\ell,
\end{equation}
such that at lowest order one obtains
\begin{equation}
   (\vec X^{-1}- \vec H)^{-1} = \vec X + \mathcal{O}((X_j - X'_j)^2). 
\end{equation}
Now, Lemma~\ref{lemma:Schur-lossy}, criterion II.~\!(b)  can be reformulated as follows: For all non-zero vectors $\vec y \in \mathbb{C}^N$ we must have
\begin{equation}
    \vec y^\dagger (\vec \Lambda + \vec\Gamma^\text{d}) \vec y > 
    \vec y^\dagger (\vec A + \vec N)^\top (\vec X^{-1}- \vec H)^{-1} (\vec A+ \vec N) \vec y.
\end{equation}
For a particular choice of $\vec y$ one obtains a necessary condition for stability. Taking $\vec y = \vec v'_F$, the above results in
\begin{equation}
   \lambda_2 > \vec {v'}_F^{\dagger} (\vec A + \vec N)^\top
      (\vec X^{-1}- \vec H)^{-1} (\vec A + \vec N) \vec v'_F ,
\end{equation}
were applying the aforementioned expansion on the right-hand side, at leading order in $(X_j - X'_j)$, yields
\begin{equation}
  \lambda_2 > \vec v_F^\dagger (\vec A + \vec N)^\top \vec X (\vec A + \vec N) \vec v_F + \mathcal{O}((X_j - X'_j)^2),
\end{equation}
taking into account that the eigenvalues of $\Lambda+\Gamma$ and $\Lambda+\Gamma^d$ coincide to leading order in the lossy case (cf. Lemma \ref{lemma:6}).
Given now the symmetries of $\vec A$, $\vec N$, and $\vec X$, one can expand the result as
\begin{equation}
\begin{aligned}
    \lambda_2 > ~&\vec v_F^\dagger \left[ \vec A^\top \vec X \vec A + 2\vec A^\top \vec X \vec N + \vec N \vec X \vec N \right] \vec v_F
     \\& \qquad\qquad\qquad\qquad\qquad\qquad+\mathcal{O}((X_j - X'_j)^2). 
\end{aligned}
\end{equation}
This concludes the proof.
This corollary entails a previous result in Ref.~\cite{Sharafutdinov2018}.
\end{proof}

Note that each term of the matrices on the right hand side of~\eqref{eq:stab-lambda-small-x} is symmetric and hence contributes positively, adding to the lower bound on the algebraic connectivity $\lambda_2$ of the system.
This implies that resistive networks \textit{always} require a higher degree of connectivity to ensure stability.

\begin{cor}\label{cor:5}
A resistive power grid needs to ensure
\begin{equation}\label{eq:cor5}
    \lambda_2 >   \sum_j (X_j - X'_j)v_{Fj}^2\left(\sum_{k\neq j}^N E_k^{\circ} G_{j,k} \right)^{\!2}\!\!,
\end{equation}
in the limiting case of no power exchange, to leading order in the losses.
\end{cor}
\begin{proof}
If there is a negligible power exchange in the power grid, all rotor angles $\delta^\circ_j, \forall j$ are identical, such that
\begin{equation}
    \cos(\delta_\ell^\circ - \delta_j^\circ) = 1, \quad
    \sin(\delta_\ell^\circ - \delta_j^\circ) = 0,
\end{equation}
for all connections $(j,\ell)$. This results in $A_{j,\ell}=0$ and $\Gamma^d_{j,j} = 0$ in~\eqref{eq:def_matrices}, and Corollary~\ref{cor:4} reads
\begin{equation}
        \lambda_2 > \vec v_F^\top \vec N \vec X \vec N \vec v_F,
\end{equation}
where all matrices are diagonal matrices.
Writing the terms explicitly yields~\eqref{eq:cor5}, entailing a lower bound to the connectivity of a power grid with resistive elements while considering only leading order of the losses.
\end{proof}

Before proceeding with the final corollaries, note that despite the cumbersome matrix notation employed here, one can still extract very useful information -- which can easily be computed numerically if desired -- by utilising different \textit{matrix norms}.
\begin{lemma}
Let $\vec Z\in \mathbb{C}^{N\times N}$ and $\vec W\in \mathbb{C}^{N\times N}$ be two matrices, and let $\| \cdot \|_{n}$ denote an $n$-induced matrix norm, one has 
\begin{equation}
    \| \vec Z \vec W \|_{n} \leq \| \vec Z \|_{n}\| \vec W \|_{n},
\end{equation}
i.e., all induced matrix norms are sub-multiplicative.
\end{lemma}
Furthermore, recall that $\| \cdot \|_{2}$ denotes the $\ell_2$-norm for vectors, also known as \textit{spectral norm} or \textit{Euclidean norm}.

\begin{cor}\label{cor:6}
If a positive algebraic connectivity $\lambda_2>0$, and all nodes $j=1,\ldots, N$, 
\begin{equation}\label{eq:stab-cond-Xl}
   (X_j\! -\! X'_j)^{-1} - \sum_{\ell=1}^N B_{j,\ell} > 
     \frac{\| (\vec A \!+\! \vec N)  \|_{2}  \| (\vec A\! +\! \vec N)^\top \|_{2}}{\lambda_2},
\end{equation}
where $\| \cdot \|_2$ is the induced $\ell_2$-norm, then an equilibrium point is linearly stable to leading order in the losses.
\end{cor}
\begin{proof}
A positive algebraic connectivity $\lambda_2>0$ implies that $\vec \Lambda + \vec \Gamma^d$ is positive definite on $\mathcal{D}^{(1)}_\perp$, and criterion I.~\!(a) in Lemma~\ref{lemma:Schur-lossy} is satisfied.

Consider now criterion I.~\!(b) in Lemma~\ref{lemma:Schur-lossy}. Using Ger\v{s}gorin's circle theorem, as in the proof of Corollary~\ref{cor:3}, one finds that condition~\eqref{eq:stab-cond-Xl} imply
\begin{equation}
    (\vec X^{-1} - \vec H ) - \lambda_2^{-1}\| (\vec A + \vec N)  \|_{2}  \| \vec (\vec A + \vec N)^\top \|_{2} \eye ,
\end{equation}
is positive definite. Noting that to leading order in the losses we have $\lambda_2^{-1}=\| (\vec \Lambda + \vec \Gamma)^+ \|_{2} =\| (\vec \Lambda + \vec \Gamma^d)^+ \|_{2} $, this implies that $\forall \vec y \in \mathbb{R}$
\begin{equation}
\begin{aligned}
    \vec y^\top& (\vec X^{-1} - \vec H) \vec y  \\
    &> \| \vec A + \vec N \|_{2} \| (\vec \Lambda + \vec \Gamma^d)^+ \|_{2} \| (\vec A +\vec N)^\top \|_{2}
         \| \vec y \|^2  \\
    & \ge \| (\vec A +\vec N) (\vec \Lambda + \vec \Gamma^d)^+ \vec (\vec A +\vec N)^\top \|_{2} \| \vec y \|^2  \\
    & \ge \vec y^\top (\vec A +\vec N) (\vec \Lambda + \vec \Gamma^d)^+ (\vec A +\vec N)^\top \vec y.
\end{aligned}
\end{equation}
Hence, matrix $\vec H - \vec X^{-1} + (\vec A +\vec N) (\vec \Lambda + \vec \Gamma^d)^+ (\vec A +\vec N)^\top$ is negative definite and criterion I.~\!(b) in Lemma~\ref{lemma:Schur-lossy} is satisfied. 
The equilibrium is linearly stable to leading order in the losses.
\end{proof}

\begin{cor}\label{cor:7}
If by criterion II.~\!(a) in Lemma~\ref{lemma:Schur-lossy} the matrix $ \vec H - \vec X^{-1}$ is negative definite, and if the algebraic connectivity $\lambda_2$ satisfies
\begin{equation}\label{eq:stab-cond-lam}
    \lambda_2 > \| (\vec A + \vec N)^\top \vec  (\vec H  - \vec X^{-1})^{-1}  (\vec A + \vec N)\|_2,
\end{equation}
where $\| \cdot \|_2$ is the induced $\ell_2$-norm, then, to leading order in the losses, the equilibrium point is linearly stable.
\end{cor}
\begin{proof} 
Assume that $ \vec H - \vec X^{-1}$ is negative definite as given by criterion II.~\!(a) in Lemma~\ref{lemma:Schur-lossy}.
The assumption~\eqref{eq:stab-cond-lam} implies that $ \forall \vec y \in \mathcal{D}^{(1)}_\perp$
\begin{equation}
\begin{aligned}
   \vec y^\top  (\vec \Lambda~+ &~ \vec \Gamma^d)\vec y \ge \lambda_2 \|\vec y\|^2  \\
   &> \| (\vec A + \vec N)^\top  (\vec H \!-\!\vec X^{-1})^{-1}  (\vec A + \vec N) \|_2 
          \|\vec y\|^2  \\
   &\ge \vec y^\top (\vec A + \vec N)^\top  (\vec H \!-\!\vec X^{-1})^{-1}  (\vec A + \vec N) \vec y, 
\end{aligned}
\end{equation}
again noticing that the eigenvalues  for $(\vec \Lambda + \vec \Gamma^d)$ and $(\vec \Lambda + \vec \Gamma)$ coincide to leading order. Thus the matrix $(\vec \Lambda + \vec \Gamma^d) + (\vec A + \vec N)^\top (\vec H - \vec X^{-1})^{-1} (\vec A + \vec N)$ is negative definite in $\mathcal{D}^{(1)}_\perp$.
Criterion II.~\!(b) in Lemma~\ref{lemma:Schur-lossy} is therefore satisfied and the equilibrium is linearly stable to leading order in the losses.
\end{proof}

\section{Numerical Analysis}\label{sec:VI}\vspace*{-2.5mm}

In this section, we present a numerical analysis on two model systems to test how tight are the bounds given by the above criteria, i.e., Corollaries~\ref{cor:1} to~\ref{cor:7}, and ultimately showcase their utility.
First, a system consisting of two machines with one acting as a generator, producing power ($P^{m} > 0$), the other acting as a motor, consuming power ($P^{m} < 0$)~\cite{Filatrella2008,Schmietendorf2016}.
Second, a system compromised of three motors and three generators, connected in a ring.
The topology and parameters are given in Fig.~\ref{fig:setup_test_systems}.

While the active power of each node and the admittance of the lines connecting each node can differ for different nodes, all other parameters (e.g., difference in reactance, damping, inertia, relaxation time, and internal voltages) are set to a single value for each machine.
To check the stability boundary we increase/decrease $P$ at all nodes proportionally. 
More precisely, we start from the base values in Fig.~\ref{fig:setup_test_systems} and multiply by $P_f$ to change both the output of generator nodes and the consumption of consumer nodes, thus increasing the load of the transmission lines.
Evaluating whether the corollaries correspond with the linear stability analysis necessitates solving for fixed points $(\mathbf{\delta}^\circ,\mathbf{\omega}^\circ,\mathbf{E}^\circ)$ given by~\eqref{eq:3rd-fixed}.
Specifically, solutions that have a vanishing angular frequency $\omega_j^\circ=0 \; \forall \; j$, since they correspond to solutions with the system operating at the desired reference frequency.
While there are multiple possible solutions of~\eqref{eq:3rd-fixed} that carry no physical meaning (e.g., having a negative voltages) we focus on one stable solution with physical meaning.

\begin{figure}
    \includegraphics[width=\columnwidth]{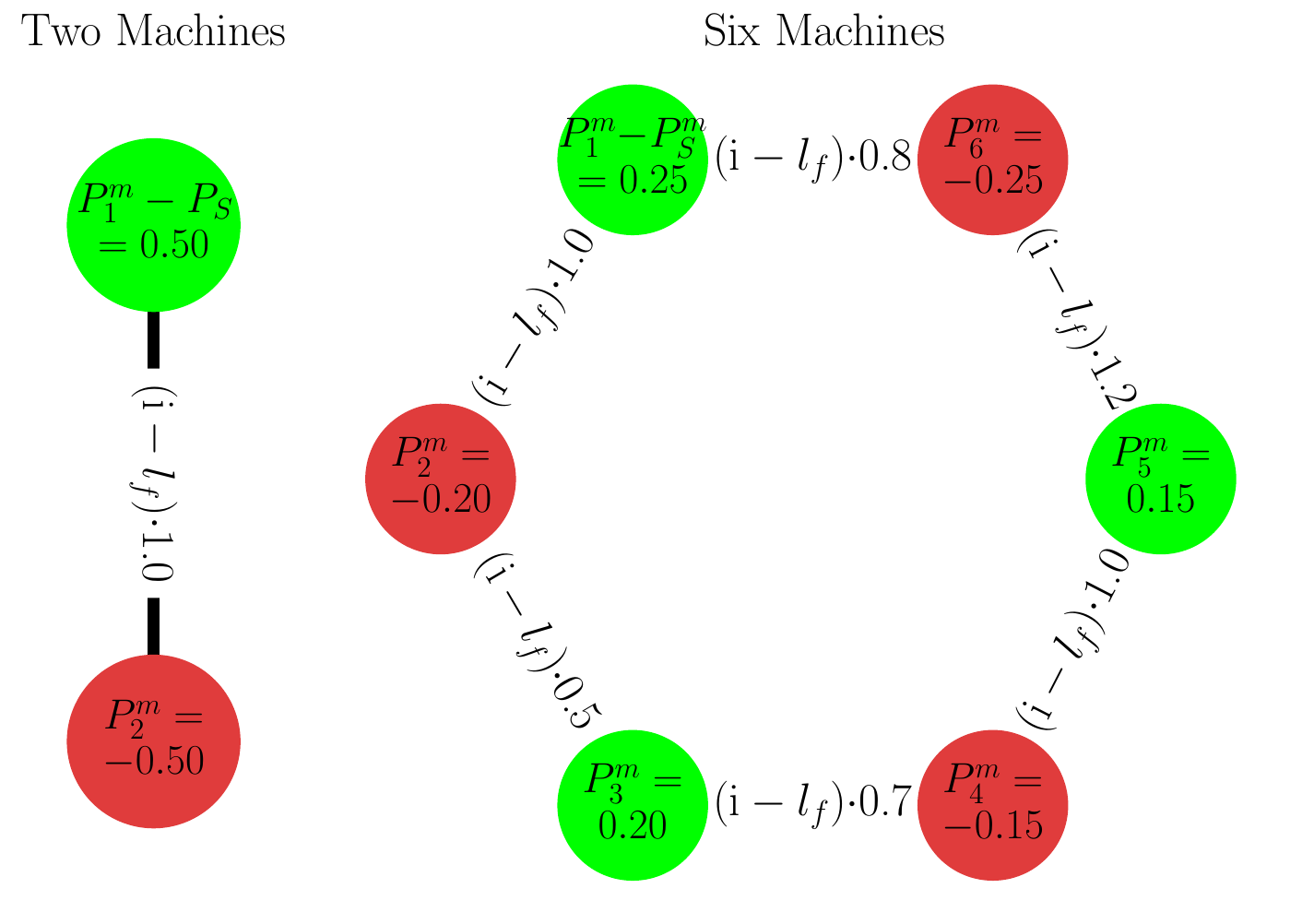}\vspace*{-2.5mm}
    \caption{Topology of the two-machine system and six-machine system that were used in the numerical study.
    Synchronous generators and motors are indicated by green and red vertex colour, respectively.
    The vertex labels show the active power at each machine $P^{m}_i$.
    The power $P^{m}_{\mathrm{S}}$ that is needed to balance the system in case of transmission losses depends on the calculated fixed point.
    The edge labels show the admittance of the lines between the machines with the real part $G_{j,l}$ that is associated with losses given by the product of the imaginary part of the admittance $B_{j,l}$ and the loss factor $l_f$.
    The shunt admittance is chosen as $B_{i,s}=0.2$ for the two-machine system and $B_{i,s}=0$ for the six-machine system.
    In both cases the damping constant $D$, inertia $M$, relaxation time $T$, and internal voltage $E_f$ were equal for each machine.
    They were chosen as $D=0.2$, $M=1$, $T=2$ and $E_f=1$.}\label{fig:setup_test_systems}\vspace*{-2.5mm}
\end{figure}

In order to find a stable state of the system, we use the double-checking procedure described below.
First, we solve the equations and find a fixed point using a root solver provided by python's SciPy package~\cite{2020SciPy-NMeth} starting from the solution of the linearised equations.
Subsequently, we perform a the two-step procedure:
\vspace*{-1.5mm}
\begin{enumerate}\itemsep0em
    \item We perform a first check of the stability of the fixed point by evaluating the eigenvalue spectrum of the associated Jacobian given in~\eqref{eq:jacobian}.
    \item We perform a second check of the stability of the fixed point by perturbing the fixed point by a small random disturbance and numerically solved the full equations in time domain using an appropriate fifth-order adaptive numerical solver~\cite{Ansmann2018}.
\end{enumerate}
\vspace*{-1.5mm}
If the fixed point found in step 1) is not linearly stable, step 2) slightly perturbs the system, forcing it to relax to a new fixed point, which we take as the final stable fixed point. 
This double-checking procedure, in contrast with merely employing the root-finding algorithm, ensures the fixed point that is found is stable.
The subsequent analysis will require varying the system's parameters.
To ensure the system remains in a stable fixed point, we change the system's parameters in small steps (adiabatically).
For sufficiently small steps, the fixed point changes only slightly and, given the system remains stable, a new fixed point can most likely be found by the root-solving algorithm when initialising the search with the previously obtained fixed point.

The tightness and therefore usefulness the corollaries is tested for the lossless case ($G_{j,\ell}=0$) as well as for increasing losses.
More precisely, we assume a fixed ratio of conductances and admittances, $G_{j,\ell} = l_f B_{j,\ell}$, for all lines $(j,\ell)$. 
By increasing the loss-factor $l_f$ and thus increasing the resistive losses, we investigate how the bounds of stability change and how they compare to the numerical results, keeping in mind that the corollaries are only correct up to leading order when considering losses.
Since we check the corollaries for different loss factors $l_f$, this approximation is therefore tested. 
For both instances we scan over a range of different active power levels (i.e., $P_1$ or $P_f$) and differences in reactance $\Delta X = X_j - X_j' = X - X'$ to find sets of parameters where a stable fixed point exists.

First, an examination of the pure voltage and pure rotor-angle stability is put forward, comparing Corollary~\ref{cor:1} and~\ref{cor:2}.
Second, the mixed instability corollaries are tested.
Corollaries~\ref{cor:4}, \ref{cor:6}, and \ref{cor:7} are tested in a similar setting as above.
To distinguish between different instabilities the eigenvalues and the eigenvectors of the associated Jacobian in~\eqref{eq:jacobian} are examined for each system in a lossless and lossy setting. 
As a fixed point becomes unstable one or multiple eigenvalues pass the imaginary axis. 
The corresponding eigenvectors show which kind of instability is present and which corollary to compare to.
\vspace*{-2.5mm}
\subsection{Lossless case}\vspace*{-2.5mm}
In the lossless case ($G_{j,\ell}=0$) the equations of motion are considerably simplified.
While different fixed points, stable and unstable, can be found, we focus on fixed points that are stable by using the aforementioned iterative procedure.
In the two-machine system the active power injected by one machine is given by $P_1^m$ and the power extracted by the other machine is $P^{m}_2=-P^{m}_1$.
In the six-machine system, the default power injection presented in Fig.~\ref{fig:setup_test_systems} was multiplied by the factor $P_f$, thus proportionally increasing the power extracted or injected at each node.

\begin{figure}[t]
    \includegraphics[width=\columnwidth]{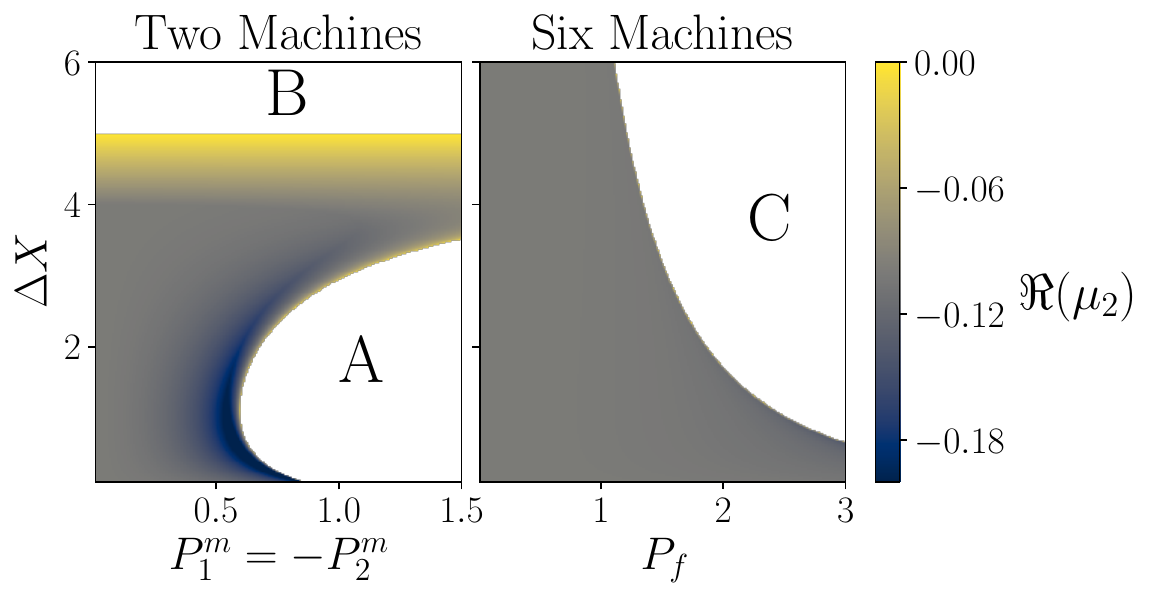}\vspace*{-2.5mm}
    \caption{Stability maps for the lossless test systems. 
    Shown is the real part of the dominant non-zero eigenvalue of the Jacobian $\mu_2$, which determines the stability of the fixed point for the two machine system (left) and the six machine system (right) as a function of levels of power injection $P_1=-P_2$ or $P_f$ for the two machine or six machine system, respectively, and the difference in reactance $\Delta X$.
    White areas depict parameter combinations where either no stable fixed point could be found or where only a nonphysical fixed point (e.g., with negative voltages) could by found.
    A and C indicate regions that are reached via mixed instability for the two and six machine system, respectively. 
    B indicates a region reached by a pure voltage instability for the two machine system.}\label{fig:no_losses_map}\vspace*{-2.5mm}
\end{figure}

In Fig.~\ref{fig:no_losses_map} we exhibit the stability map for varying $\Delta X$ and $P_1$ or $P_f$.
Here, the real part of the dominant eigenvalue is shown by the filled contours. 
White regions indicate parameters where either no stable fixed point existed or were only non-physically meaningful fixed points (e.g., with negative voltages) could be found.
Before examining the usefulness of the corollaries, we identified which type of instability corresponds to which border of the regions indicated by the letters in Fig.~\ref{fig:no_losses_map}.
We evaluated the eigenvectors corresponding to the leading eigenvalue to identify whether a pure or mixed instability was observed. 
For the two-machine system, we can observe a mixed instability at the border to region A, with the notable exception of the line with $\Delta X=0$ discussed below, and a pure voltage instability for region B.
In contrast, for the six-machine system only a mixed instability could be observed at the border to region C in Fig.~\ref{fig:no_losses_map}.
\vspace*{-5.5mm}
\subsubsection{Pure instability}\label{sec:no_losses_pure}\vspace*{-2.5mm}
The pure voltage instability that arose in the two-machine system by crossing the border into region B in\\

\noindent Fig.~\ref{fig:no_losses_map} was examined by setting $P^{m}_1=-P^{m}_2=0.5$ and varying the difference in reactance $\Delta X$ for the two-machine system.
At $\Delta X=5$ the voltages at each machine diverge and no physically meaningful stable fixed point could be found for larger $\Delta X$ (cf. Fig.~\ref{fig:no_losses_voltage}).
Corollaries \ref{cor:1} and \ref{cor:2} predict this point correctly as the left-hand and right-hand side of each corollary were equal for $\Delta X=5$.

\begin{figure}[t]
    \includegraphics[width = \columnwidth]{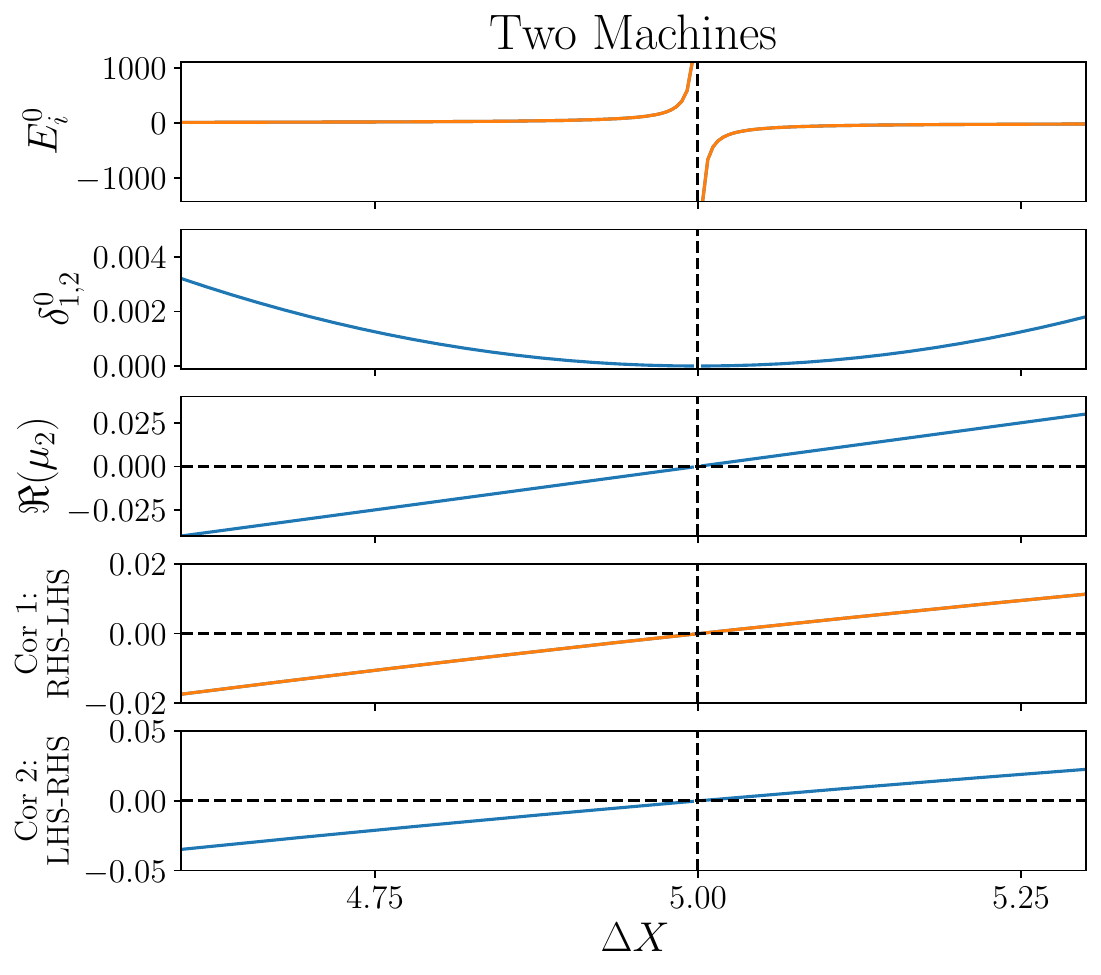}\vspace*{-2.5mm}
    \caption{Route to a pure voltage instability for the two-machine system. 
    Power injection was set to $P^{m}_1=-P^{m}_2=0.5$ and the reactance $\Delta X$ was varied.
    From top to bottom: Stationary voltages $E_i^\circ$, stationary angle differences $\delta_{i,1}$, real part of the dominant non-zero eigenvalue $\mu_2$, difference of left hand sided and right hand side of Corollary~\ref{cor:1} and Corollary~\ref{cor:2}. 
    The stable fixed pint is lost at $\Delta X = 5$, which is perfectly predicted by Corollaries \ref{cor:1} and \ref{cor:2}.}\label{fig:no_losses_voltage}\vspace*{-2.5mm}
\end{figure}

\begin{figure}[t]
    \includegraphics[width = \columnwidth]{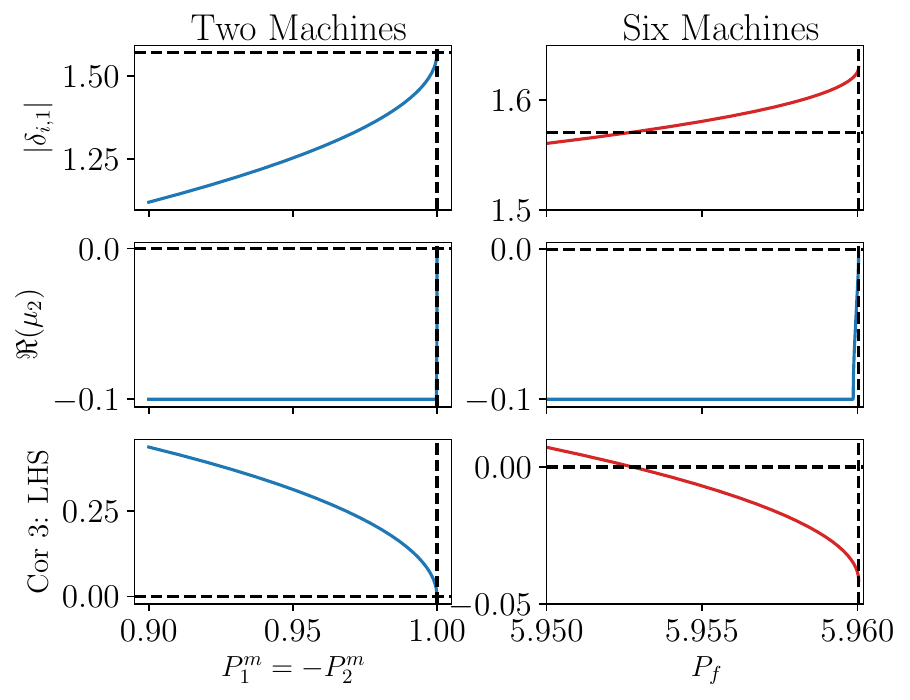}\vspace*{-2.5mm}
    \caption{Route to a pure angle instability with a difference in transient and static reactance $\Delta X = 0$ for the two machine (left column) and six machine system (right column).
    The rows show from top to bottom: Stationary phase angle difference $|\delta_{i,1}|$, real part of the dominant eigenvalue $\mu_2$ and left hand side of Corollary~\ref{cor:3}.
    The stable fixed pint is lost at $P^{m}_1=-P^{m}_2=1$ and $P_f=5.96$ for the two- and six-machine system, respectively. 
    These bifurcation points are almost perfectly predicted by Corollary~\ref{cor:3}.}\label{fig:no_losses_angle}\vspace*{-2.5mm}
\end{figure}

A pure rotor-angle instability could be observed for both systems by setting the difference of static and transient reactance to $\Delta X=0$, thus isolating the rotor-angle subsystem. 
The instability arose after increasing the level of power injection/extraction beyond $P^{m}_1=-P^{m}_2=1$ and $P_f\approx 5.96$ for the two machine and the six machine system, respectively (see Fig.~\ref{fig:no_losses_angle}).
Corollary~\ref{cor:3} predicted the point where the maximal phase angle difference $\delta_{i,j}$ of machines connected by a line was equal to $\pi/2$ for the two and six-machine system.
This coincided with the point where the stability of fixed point is lost for the two machine system, while it was slightly below the transition for the six machine system.
We conclude that the bound in Corollary~\ref{cor:3} is tight in simple systems and remains near the transition point for more complex systems.

\subsubsection{Mixed Instability}\label{sec:no_losses_mixed}\vspace*{-2.5mm}
To evaluate the usefulness of the corollaries for the mixed instability we study the left-hand and right-hand sides of Corollaries \ref{cor:4} to \ref{cor:7} for a range of values of power injection/extraction (i.e., $P^{m}_1 = -P^{m}_2$ and $P_f$) and differences in reactance $\Delta X$.
The results can be seen in Fig.~\ref{fig:no_losses_cor_map_mix}.
These plots show the difference of the left-hand side (LHS) and the right-hand side (RHS) of Corollaries \ref{cor:4} to \ref{cor:7} as a function of power levels and difference in reactances.
The fixed point is stable for both systems where LHS$>$RHS.
Corollaries \ref{cor:4}--\ref{cor:7} provide sufficient criteria for determining the linear stability of the considered fixed point.
Hence the difference LHS-LHS$>0$ ensures stability.
In contrast, if LHS$-$RHS$<0$ the corollaries yield no definitive answer.
We find that in the white region in Fig.~\ref{fig:no_losses_cor_map_mix} LHS$-$RHS$<0$ and no stable fixed point could be found, while in the gray hatched region we have LHS$-$RHS$<0$ and there is a stable fixed point.
Corollary~\ref{cor:4} seems to be tight in both cases tested numerically as the stability boundary coincides with the limiting case LHS$=$RHS.
Corollaries \ref{cor:5} and \ref{cor:6} are tight only for the two machine system. 
Nevertheless, they adequately reproduce the qualitative shape of the stability boundary for the six machine system.

\begin{figure}[t]
    \includegraphics[width = \columnwidth]{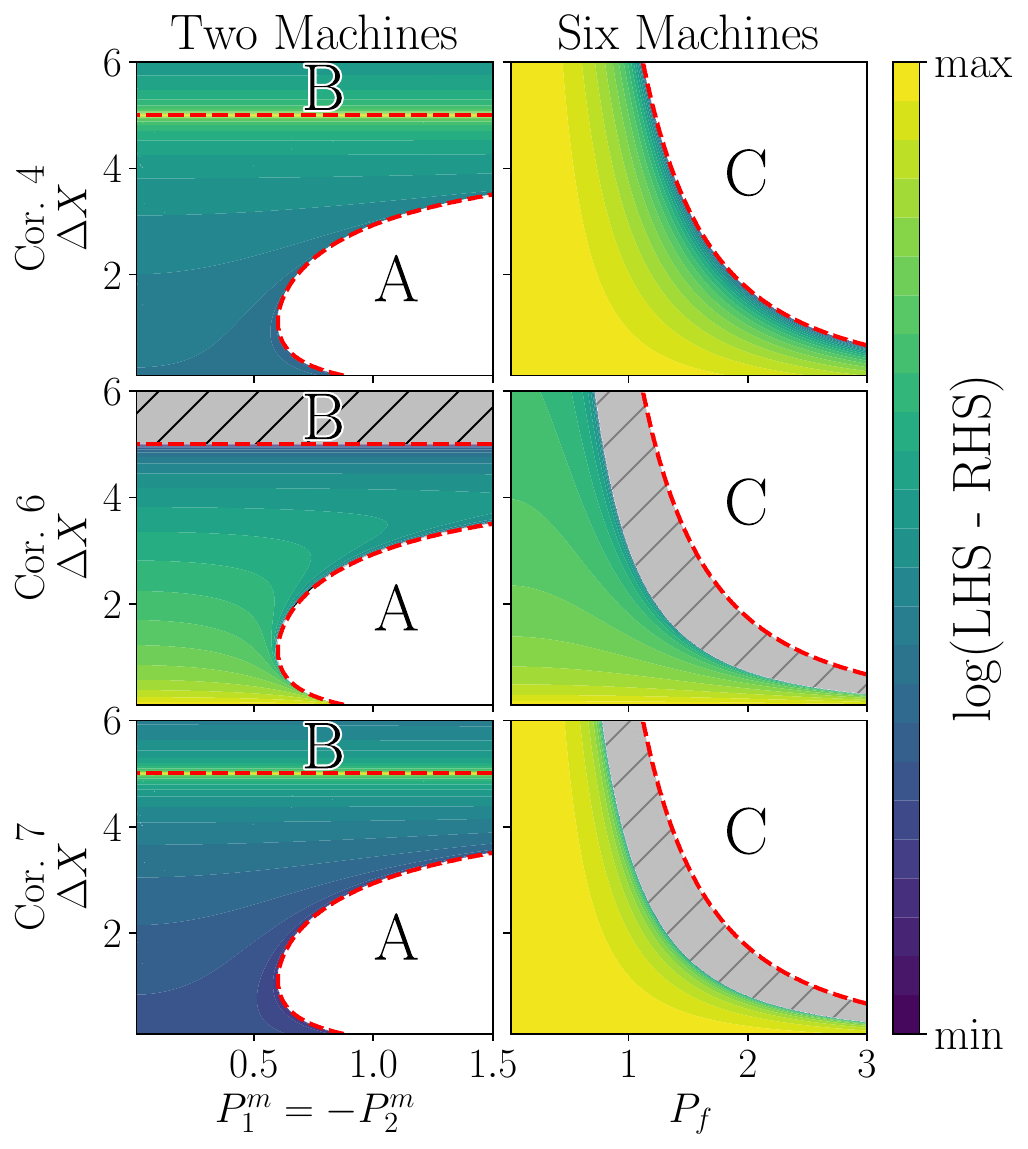}\vspace*{-2.5mm}
    \caption{Comparison of the numerically determined stability boundaries and the sufficient criteria given by Corollaries \ref{cor:4}, \ref{cor:6}, and \ref{cor:7}.
    The red dashed lines show the boundaries of the parameter regions for which a stable fixed point exists according to the dominant eigenvalue of the Jacobian $\mu_2$ in both the two- and six-machine system (cf. Fig.~\ref{fig:setup_test_systems}).
    The colourmap shows the logarithm of left-hand side (LHS) minus right-hand side (RHS) of the corollaries corresponding to mixed instabilities (i.e., Corollaries \ref{cor:4}, \ref{cor:6}, and \ref{cor:7} for two machines (left column) and six machines (right column)) as a function of power injection/extraction and difference in reactance. 
    Note that only values are shown were the difference is positive and thus the logarithm gives a real value.
    In the white region, no stable fixed points exist.
    Correspondingly, the sufficient criteria are not satisfied (LHS$-$RHS$<0$).
    In the coloured area, LHS$-$RHS$>0$, such that a stable fixed point exists according to the Corollaries \ref{cor:4}, \ref{cor:6}, and \ref{cor:7}.
    While the sufficient criteria are not satisfied (LHS$-$RHS$<0$) in the gray hatched area, a stable fixed point still exists in this area.
    }\label{fig:no_losses_cor_map_mix}\vspace*{-2.5mm}
\end{figure}

\subsection{Lossy case}\vspace*{-2.5mm}
After having considered lossless systems and showingthat corollaries \ref{cor:1} to \ref{cor:3} are tight, while corollaries \ref{cor:4} to \ref{cor:7} adequately describe the stability boundary (perfectly in the two-machine system), we turn to the more interesting case where losses are included.
Again, the two- and six-machine systems were considered.
Before checking the corollaries and how the results obtained by the perturbations Ansatz~\eqref{eq:perturb} compare to numerical results, we have to find the correct fixed points.
While a balanced system without losses can easily be obtained by choosing the active power at each node to obey $\sum_j P^{m}_j = 0$, this is not the case when considering losses, i.e., $G_{j,l} \neq 0 \, \forall j,\ell$.
The losses $P_{\mathrm{L}}$ have to be covered so that the $\sum_j P^{m}_j = P_{\mathrm{L
}}$.
Finding a fixed point determines the stationary power flows and thus the losses. 
To have a balanced system, one needs to choose a method that ensures that the overall power balance is obeyed.
Commonly, in conventional power flow studies, losses are compensated at a single bus by changing the power output at the connected machines~\cite{Machowski2020}.
To keep the characteristics of a synchronous machine for each node, we compensate the power transmission losses by modifying the power injection of the slack node.
Thus, the power conservation law~\eqref{eq:3rd-fixed} is not obeyed at the slack node, leading to a solution where the missing active power $P_{\mathrm{S}}$ is calculated and added to the power output at the slack node.
We chose the first generator bus to provide this additional power.
This is an arbitrary choice that deserves more attention in realistic simulations but is appropriate for the analysis at hand.
Its active power output is thus increased and set to $P^m_1 = P_{1}^{n} + P_{\mathrm{S}}$ with the nominal power output $P_1^n = - P^{m}_2$.
As described in Sec.~\ref{sec:II}, the real elements of the nodal admittance matrix $G_{\ell, j}$ are generally negative.
We chose to introduce different levels of losses by introducing the loss factor $l_f$ and setting $G_{\ell, j} = - l_f \cdot B_{\ell, j}$.
Since $G_{\ell,j}$ and $B_{\ell, j}$ have a similar magnitude for distribution grid, while $G_{\ell,j}$ is negligible for transmission grids, we chose to set $l_f \in [0, 0.6]$.
A parameter scan over a range for the active power injection/extraction and difference in reactance for a loss factor of $l_f=0.3$ was performed.
The resulting stability map is shown in Fig.~\ref{fig:losses_map}.
While the general shape of the parameter region with a stable fixed is the same, the border to region A moved to smaller levels of power for both systems.

\begin{figure}
    \centering
    \includegraphics[width=\columnwidth]{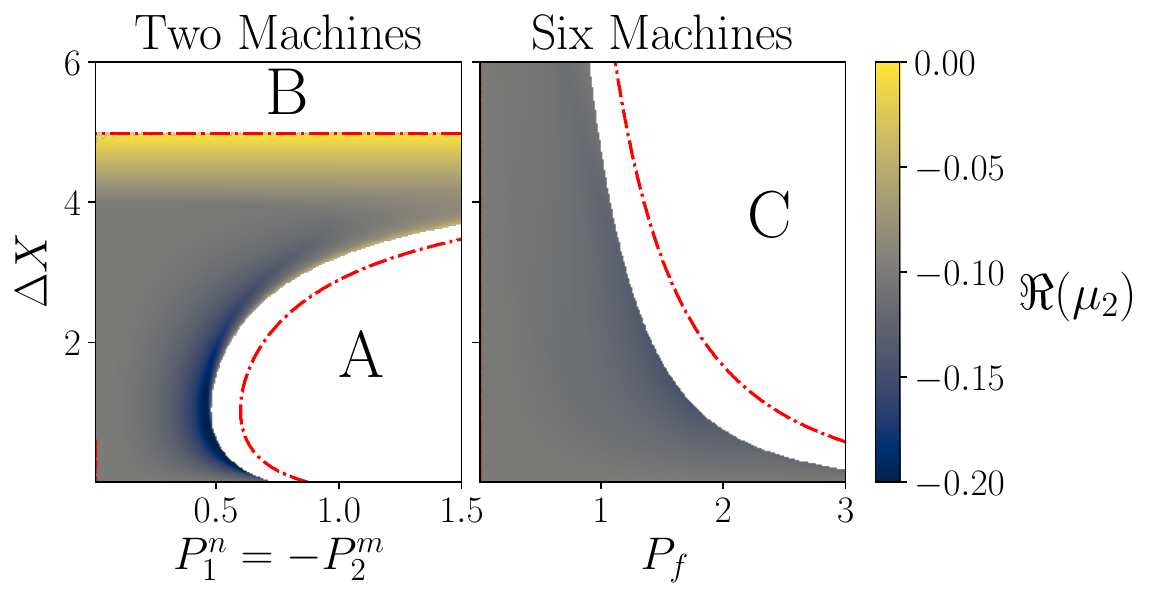}\vspace*{-2.5mm}
    \caption{Parameter region for a stable fixed point decreases slightly by shifting the mixed instability to lower levels of power injection/extraction.
    Shown are the real part of the dominant eigenvalue of the Jacobian as a function of different levels of power injection/extraction and differences in reactance for the two machine and six machine system for a loss factor of $l_f=0.3$.}\label{fig:losses_map}\vspace*{-2.5mm}
\end{figure}

If losses were considered, only mixed instability could be observed.
Therefore, the corollaries corresponding to mixed instabilities were evaluated and are shown in Fig.~\ref{fig:losses_cor_map}.
We find that the border to region A changed to lower values of power injection/extraction as indicated by the black arrows.
The filled areas show that the region where Corollaries~\ref{cor:4} to \ref{cor:7} are satisfied and do not fully cover the parameter region with a linearly stable fixed point as the loss factor $l_f$ increased.
Additionally, the approximation $\Delta X \approx 0$ used to find the Fiedler vector in Corollaries~\ref{cor:4} and \ref{cor:7} limits the range of $\Delta X$ where the corollaries are insightful in the sense of overlapping with the area of a stable fixed point given by the dominant eigenvalues.
Overall, the sufficient stability conditions remain tight for all values of $l_f$ although they were derived solely for leading order in the losses.
That is, at every point in parameter space where Corollaries~\ref{cor:4}, \ref{cor:6}, and \ref{cor:7} imply the fixed point is stable agrees with the dominant non-zero eigenvalue of the corresponding Jacobian being smaller than zero.

\begin{figure}
    \includegraphics[width=1\columnwidth]{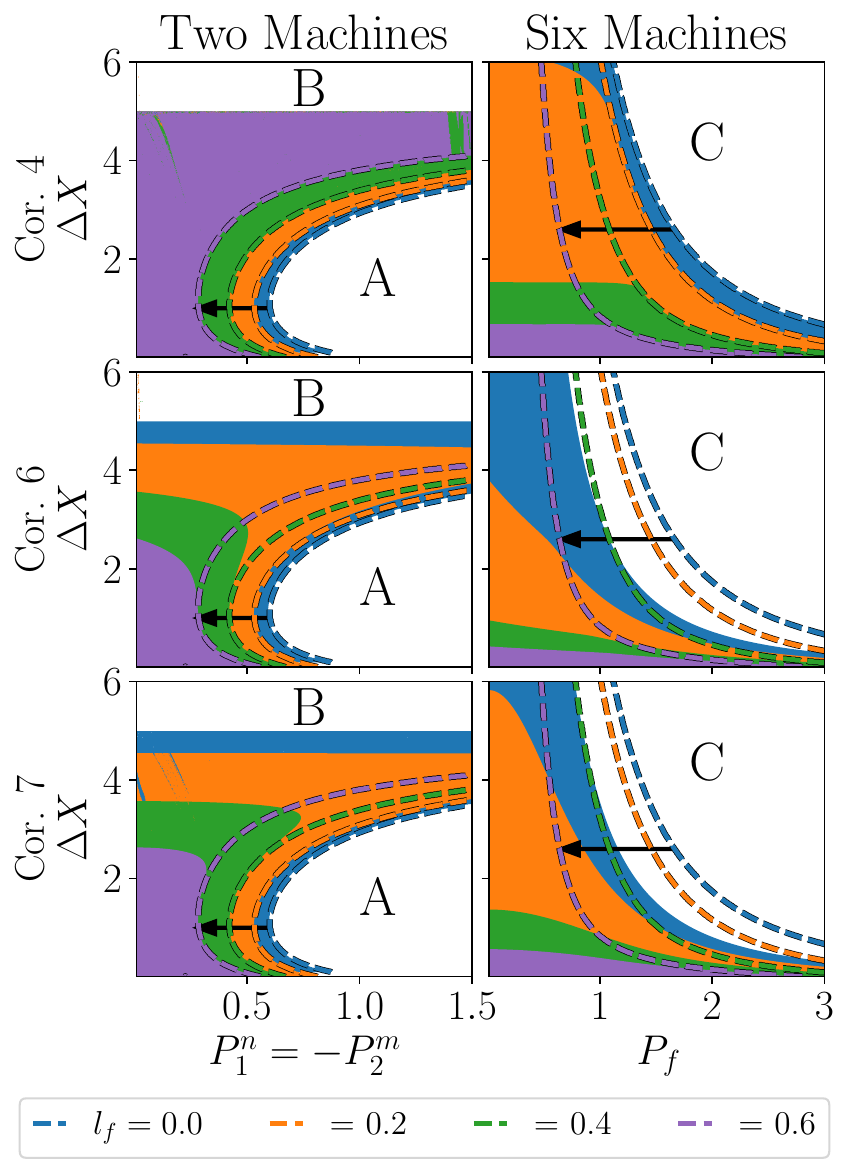}\vspace*{-2.5mm}
    \caption{Corollaries for mixed instabilities for different levels of losses in the two machine (left column) and six machine system (right column). 
    Different colours indicate different loss factors $l_f$ according to the legend below the plot. 
    Dashed lines show where the dominant non-zero eigenvalue $\mu_2$ crossed the imaginary axis. 
    The contour showing the border to region B at $\Delta X \approx 5$ in the two machine system was removed. 
    Given the hardship of following the correct fixed point close to the transitions, not all contours can be drawn with absolute accuracy.
    Nevertheless, the white regions give the parameter region where no physically meaningful stable fixed point exists.
    Region B does not change for different values of $l_f$ as already seen for $l_f=0.3$ in Fig.~\ref{fig:losses_map}.
    The borders to region A move to lower levels of power injection/extraction for higher $l_f$ as highlighted by the black arrows.
    Coloured areas show where corollaries indicate that the obtained fixed point is stable.}
    \label{fig:losses_cor_map}\vspace*{-2.5mm}
\end{figure}

Therefore, the developed Corollaries \ref{cor:1}--\ref{cor:7} can be used to efficiently judge whether a system's fixed point is stable without the need to calculate the full Jacobian or run simulations.
This is especially useful to system operator that need to check the stability of different grid situations, since they can use the corollaries to focus on the cases where the corollaries do not indicate a stable fixed point, potentially cutting down on the amount of costly (numerical) simulations of the full dynamics.

\section{Conclusion}\label{sec:VII}\vspace*{-2.5mm}
The third-order model describes the dynamics of synchronous machines and takes into account both the rotor-angle and the voltage dynamics.
Analytical results for the dynamics and the stability of coupled machines in power grids with complex topologies are rare, in particular if Ohmic losses are taken into account.
In this article, a comprehensive linear stability analysis was carried out and several explicit stability criteria were derived.

The first main result of this works depicts the influence of resistive terms of the system after linear stability analysis.
Remarkably, these terms enter into the reduced system Jacobian only via the two diagonal matrices $\vec \Gamma^d$ and $\vec N$, as shown in~\eqref{eq:def-Xi-lossy} up to leading order in the losses.
As a second main result, a decomposition of the Jacobian into the rotor-angle and the voltage subsystems is derived in Lemma~\ref{lemma:Schur-lossy}, where losses are incorporated up to linear order via perturbation theory.
This decomposition reveals clearly how the interplay of both subsystems can lead to mixed forms of instability and thus requires additional security margins. 

Based on this decomposition, several explicit stability conditions are uncovered, both for the isolated subsystems as well as for the full systems, including rotor-angle and voltage dynamics.
In particular, one can show that voltage stability is not affected directly by resistive terms up to leading order in the losses, thus implying that studies on voltage stability can be withstood in the purely lossless case.
Furthermore, Corollaries~\ref{cor:4} and \ref{cor:5} entail a strict minimum connectivity of the power-grid network solely by the presence of resistive terms, i.e., a lower bound to possible dynamics on the system given the presence of losses in the system.

The analytical insights unveiled here -- and in particular the mathematical evaluation of lossy systems -- can prove relevant to further understand power grids of all spatial scales and of general graph constructions.
By mathematically tackling the presence of losses in the system the applicability of the results is now extended from transmission grids with negligible losses to grids where loses play a bigger role.
Moreover, it opens the door to further research on higher-order models from a mathematical point-of-view, and can henceforth be applied more generally to other power-grid models.

\begin{acknowledgments}\vspace*{-2.5mm}
We thank Christopher Kaschny for helpful discussions.
We gratefully acknowledge support from the German Federal Ministry of Education and Research (grant no.~03EK3055B), the German Federal Ministry for Economic Affairs and Energy (BMWi) via the project DYNAMOS (grant no.~03ET4027A) and the Helmholtz Association via the grant \textit{Uncertainty Quantification -- From Data to Reliable Knowledge (UQ)} (grant no.~ZT-I-0029).
\end{acknowledgments}

\bibstyle{apsrev4-2}
\bibliography{bib}

\end{document}